\documentclass[submission,copyright,creativecommons]{eptcs}
\providecommand{\event}{QPL 2023}

\usepackage{amsmath,amsthm,amssymb}
\usepackage{xspace,enumerate,color,epsfig}
\usepackage{float}
\usepackage{graphicx}

\usepackage{subfiles}

\usepackage{stmaryrd}
\usepackage{mathtools}
\usepackage{microtype}
\usepackage{multirow}
\usepackage{url}
\usepackage{hyperref}
\usepackage{caption}
\usepackage{enumitem}
\usepackage{multicol}
\usepackage{todonotes}

\usepackage[utf8]{inputenc}
\usepackage{scrextend}
\usepackage[english]{babel}
\usepackage{xurl}

\usepackage{xcolor}
\definecolor{zx_green}{rgb}{216,248,216}
\definecolor{zx_red}{rgb}{232,165,165}

\usepackage{tikzit}

\tikzstyle{dot}=[inner sep=0.3mm, minimum width=2mm, minimum height=2mm, draw, shape=circle, font={\footnotesize}, tikzit fill=magenta]
\tikzstyle{white dot}=[dot, fill=white, text depth=-0.2mm, tikzit category=ZH-pf]
\tikzstyle{gray dot}=[dot, fill=black, text depth=-0.2mm, tikzit category=ZH-pf]
\tikzstyle{gray phase dot}=[gray dot, fill=black, tikzit fill=magenta]
\tikzstyle{hadamard}=[fill=yellow, draw, inner sep=0.6mm, minimum height=1.5mm, minimum width=1.5mm, shape=rectangle, tikzit shape=rectangle, tikzit category=ZH-pf]
\tikzstyle{small hadamard}=[fill=yellow, draw, inner sep=0.6mm, minimum height=1.5mm, minimum width=1.5mm, tikzit shape=rectangle]
\tikzstyle{small gray hadamard}=[fill=black, draw, inner sep=0.6mm, minimum height=1.5mm, minimum width=1.5mm, tikzit shape=rectangle]
\tikzstyle{halfscalar}=[star, fill=black, draw=black, minimum size=6pt, inner sep=0pt]
\tikzstyle{scalar}=[regular polygon, fill=black, draw=black, minimum size=6pt, inner sep=0pt, regular polygon sides=3]
\tikzstyle{box}=[shape=rectangle, text height=1.5ex, text depth=0.25ex, yshift=0.2mm, fill=white, draw=black, minimum height=3mm, minimum width=5mm, font={\small}]
\tikzstyle{smallbox}=[draw, fill=white, inner sep=0.6mm, minimum height=1.5mm, minimum width=1.5mm, font={\scriptsize}, shape=rectangle]
\tikzstyle{multiply box}=[fill=white, draw, inner sep=0.6mm, minimum height=1ex, minimum width=2ex, text depth=0ex, text height=1.25ex, tikzit shape=rectangle, font={$\cdot$}]
\tikzstyle{Z dot}=[inner sep=0mm, minimum size=2mm, shape=circle, draw=black, fill={rgb,255: red,216; green,248; blue,216}, tikzit fill={rgb,255: red,216; green,248; blue,216}]
\tikzstyle{Z phase dot}=[minimum size=5mm, font={\footnotesize\boldmath}, shape=rectangle, rounded corners=2mm, inner sep=0.2mm, outer sep=-2mm, scale=0.8, tikzit shape=circle, draw=black, fill={zx_green}, tikzit fill={rgb,255: red,216; green,248; blue,216}, tikzit draw=blue]
\tikzstyle{X dot}=[Z dot, shape=circle, draw=black, fill={zx_red}]
\tikzstyle{X phase dot}=[Z phase dot, tikzit shape=circle, tikzit draw=blue, fill={zx_red}, font={\footnotesize\color{black}\boldmath}]
\tikzstyle{lbl}=[font={\scriptsize}]
\tikzstyle{monoid}=[shape=semicircle, fill=white, draw=black, inner sep=0.5mm, rotate=180]
\tikzstyle{gray monoid}=[shape=semicircle, fill=black, draw=black, inner sep=0.5mm, rotate=180]
\tikzstyle{red monoid}=[shape=semicircle, fill={zx_red}, draw=black, inner sep=0.5mm, rotate=180]
\tikzstyle{yellow monoid}=[shape=semicircle, fill=yellow, draw=black, inner sep=0.5mm, rotate=180]
\tikzstyle{comonoid}=[shape=semicircle, fill=white, draw=black, inner sep=0.5mm]
\tikzstyle{gray comonoid}=[shape=semicircle, fill=black, draw=black, inner sep=0.5mm]
\tikzstyle{red comonoid}=[shape=semicircle, fill={zx_red}, draw=black, inner sep=0.5mm]
\tikzstyle{yellow comonoid}=[shape=semicircle, fill=yellow, draw=black, inner sep=0.5mm]
\tikzstyle{green bra}=[fill={rgb,255: red,216; green,248; blue,216}, draw=black, regular polygon, regular polygon sides=3, tikzit fill={rgb,255: red,216; green,248; blue,216}, tikzit draw=black, inner sep=0 mm, outer sep=0 mm, shape border rotate=0, font={\tiny}, minimum height=9mm]
\tikzstyle{green ket}=[fill={rgb,255: red,216; green,248; blue,216}, draw=black, regular polygon, regular polygon sides=3, tikzit fill={rgb,255: red,216; green,248; blue,216}, tikzit draw=black, inner sep=0 mm, outer sep=0 mm, shape border rotate=180, font={\tiny}, minimum height=9mm]
\tikzstyle{ket}=[draw=black, regular polygon, regular polygon sides=3, tikzit draw=black, inner sep=0 mm, outer sep=0 mm, shape border rotate=180, font={\tiny}, minimum height=9mm]
\tikzstyle{gn}=[inner sep=0mm, minimum size=2mm, shape=circle, draw=black, fill={rgb,255: red,216; green,248; blue,216}, tikzit fill={rgb,255: red,216; green,248; blue,216}]
\tikzstyle{srn}=[Z dot, shape=circle, draw=black, fill={rgb,255: red,232; green,165; blue,165}, tikzit fill={rgb,255: red,232; green,165; blue,165}, minimum size=1mm]
\tikzstyle{rn}=[Z dot, shape=circle, draw=black, fill={rgb,255: red,232; green,165; blue,165}, tikzit fill={rgb,255: red,232; green,165; blue,165}]
\tikzstyle{pn}=[Z dot, fill={rgb,255: red,255; green,200; blue,240}, tikzit fill={rgb,255: red,255; green,200; blue,240}]
\tikzstyle{yn}=[Z dot, fill=yellow, tikzit fill=yellow]
\tikzstyle{H box}=[rectangle, fill=yellow, draw=black, xscale=1, yscale=1, font={\small}, inner sep=0.75pt, minimum width=0.15cm, minimum height=0.15cm, tikzit shape=rectangle]
\tikzstyle{yellow hadamard}=[fill=yellow, draw=black, shape=rectangle, inner sep=0.6mm, minimum height=1.5mm, minimum width=1.5mm]
\tikzstyle{ug}=[regular polygon, regular polygon sides=3, fill={zx_red}, draw=black, inner sep=0pt, minimum width=1em, tikzit draw=blue]
\tikzstyle{st}=[star, star points=5, fill=white, draw=black, inner sep=1.2pt, line width=1.2pt, tikzit fill=blue, tikzit draw=red, tikzit category=ZH-pf]
\tikzstyle{not}=[fill={rgb,255: red,180; green,180; blue,180}, draw=black, shape=circle, font={$\neg$}, dot]
\tikzstyle{bbindex}=[font={\color{blue}\footnotesize}]
\tikzstyle{wide point}=[fill=white, draw, shape=isosceles triangle, shape border rotate=-90, isosceles triangle stretches=true, inner sep=0pt, minimum width=1.5cm, minimum height=6.12mm, yshift=-0.0mm]
\tikzstyle{medium gray box}=[semilarge box, fill={rgb,255: red,180; green,180; blue,180}]
\tikzstyle{small box}=[rectangle, inline text, fill=white, draw, minimum height=5mm, yshift=-0.5mm, minimum width=5mm, font={\small}]
\tikzstyle{small gray box}=[small box, fill={rgb,255: red,180; green,180; blue,180}]
\tikzstyle{medium box}=[rectangle, inline text, fill=white, draw, minimum height=5mm, yshift=-0.5mm, minimum width=8mm, font={\small}]
\tikzstyle{wire label}=[font={\footnotesize}, tikzit fill=blue, anchor=west, shape=rectangle, inner sep=1pt, xshift=-1mm]
\tikzstyle{0 control}=[minimum size=1mm, shape=circle, draw=black, fill=white, font={\footnotesize\boldmath}, inner sep=0.5pt]
\tikzstyle{Xplus dot}=[Z dot, shape=circle, draw=black, fill=red]
\tikzstyle{neg0 control}=[0 control, inner sep=0 mm, fill=white, draw=black]
\tikzstyle{hz}=[small hadamard, fill={rgb,255: red,216; green,248; blue,216}, shape=rectangle, tikzit fill={rgb,255: red,216; green,248; blue,216}, minimum height=1.5 mm, minimum width=0.75 mm, tikzit draw=black, text width=0.1 mm, inner sep=0.1 mm]
\tikzstyle{hx}=[small hadamard, fill=red, shape=rectangle, tikzit fill=red, minimum height=1.5 mm, minimum width=0.75 mm, tikzit draw=black, text width=0.1 mm, inner sep=0.1 mm]
\tikzstyle{ctrl}=[inner sep=0mm, minimum size=1mm, shape=circle, draw=black, fill=black]
\tikzstyle{targ}=[draw=black, fill=white, font={\footnotesize\boldmath}, minimum size=0.5mm, inner sep=-0.5mm, shape=circle, tikzit shape=circle, tikzit draw=black]
\tikzstyle{XD dot}=[shape=XDdot, inner sep=2pt, draw=black, tikzit fill={rgb,255: red,255; green,191; blue,191}]
\tikzstyle{XD phase dot}=[shape=XDdotphase, minimum size=4.75mm, font={\footnotesize}, inner sep=.1mm, outer sep=0mm, scale=0.8, tikzit shape=circle, rounded corners=1.9mm, draw=black, tikzit fill={rgb,255: red,255; green,191; blue,191}, tikzit draw=blue]

\tikzstyle{gray}=[-, draw={blue!60!white}, tikzit draw=blue]
\tikzstyle{blue}=[-, draw={blue!60!white}, tikzit draw=blue]
\tikzstyle{brace edge}=[-, tikzit draw=blue, decorate, decoration={brace,amplitude=1mm,raise=-1mm}]
\tikzstyle{diredge}=[->]
\tikzstyle{not edge}=[-, dashed, dash pattern=on 2pt off 1.5pt, thick, draw={rgb,255: red,255; green,68; blue,68}]
\tikzstyle{thin}=[-, line width=0.10mm]

\input{zh.tikzdefs}

\newcommand{\id}{\text{id}}

\theoremstyle{definition}
\newtheorem{theorem}{Theorem}[section]
\newtheorem{corollary}[theorem]{Corollary}
\newtheorem{lemma}[theorem]{Lemma}
\newtheorem{proposition}[theorem]{Proposition}

\newtheorem{example*}[theorem]{Example*}
\newtheorem{examples*}[theorem]{Examples*}
\newtheorem{remark}[theorem]{Remark}
\newtheorem{remark*}[theorem]{Remark*}

\newcommand{\citehack}{}

\title{The Qudit ZH-Calculus: \\Generalised Toffoli+Hadamard and Universality}
\def\titlerunning{The Qudit ZH-Calculus}

\author{Patrick Roy\footnote{Work completed before the author joined Amazon UK, Inc.}
\institute{University of Oxford}
\email{patrick.roy@udo.edu} \and
John van de Wetering
\institute{University of Amsterdam}
\email{john@vdwetering.name} \and
Lia Yeh
\institute{University of Oxford}
\institute{Quantinuum, 17 Beaumont Street\\Oxford OX1 2NA, United Kingdom}
\email{lia.yeh@cs.ox.ac.uk} 
}
\def\authorrunning{P.~Roy, J.~van de Wetering \& L.~Yeh}

\begin{document}

\maketitle

\begin{abstract}
    We introduce the qudit ZH-calculus and show how to generalise the phase-free qubit rules to qudits. We prove that for prime dimensions $d$, the phase-free qudit ZH-calculus is universal for matrices over the ring $\mathbb{Z}[e^{2\pi i/d}]$. 
    For qubits, there is a strong connection between phase-free ZH-diagrams and Toffoli+Hadamard circuits, a computationally universal fragment of quantum circuits. 
    We generalise this connection to qudits, by finding that the two-qudit $\ket{0}$-controlled $X$ gate can be used to construct all classical reversible qudit logic circuits in any odd qudit dimension, which for qubits requires the three-qubit Toffoli gate. We prove that our construction is asymptotically optimal up to a logarithmic term.
    Twenty years after the celebrated result by Shi proving universality of Toffoli+Hadamard for qubits, we prove that circuits of $\ket{0}$-controlled $X$ and Hadamard gates are approximately universal for qudit quantum computing for any odd prime $d$, and moreover that phase-free ZH-diagrams correspond precisely to such circuits allowing postselections.
\end{abstract}

\section{Introduction}

For qubits there are essentially three different graphical calculi: ZX, ZW and ZH~\cite{carette2020}. Each of these is suitable for reasoning about different types of structures and quantum gates. The ZX-calculus~\cite{CD1,CoeckeB2011interacting} is the most well-studied of these, and can naturally reason about the Clifford+Phases gate set (containing CNOT, Hadamard, $S$ as well as arbitrary $Z$ phase gates) and the useful primitives of phase gadgets and Pauli gadgets~\cite{phaseGadgetSynth,wetering2020}. Its \emph{phase-free} fragment, where the spiders cannot be labelled by a non-trivial phase, corresponds to CNOT circuits (together with ancillae and postselection) and can alternatively be interpreted into a category of linear relations~\cite{kissingerCSS}.
The ZW-calculus~\cite{hadzihasanovic2017, wang2021} instead can reason about photonic and fermionic computations~\cite{hadzihasanovic2018diagrammatic}. The W-spider helps to easily represent sums of linear maps~\cite{Koch2022Masters,shaikh2022sum,wang2022differentiating}. Its phase-free fragment is universal and complete for matrices over $\mathbb{Z}$, and here again the W-spider is used to sum up numbers.

The calculus we will be interested in here is the ZH-calculus~\cite{Backens2019ZH,backens2021}. Its H-box generator allows for easy representation of gates involving multilinear logic, like the Toffoli or other many-controlled gates. It can represent hyper-graph states~\cite{KupperThesis}, the path-sum formalism~\cite{lemonnier2020,pathsRenaud,Vilmart2022Completeness}, quantum binary decision diagrams~\cite{vilmart2021quantum} and more~\cite{d.p.east2021spinnetworks,east2022}. Its phase-free fragment represents the Toffoli+Hadamard gate set and is universal for matrices over $\mathbb{Z}$~\cite{backens2021}. The H-box here allows for representing the AND operation $\ket{x}\otimes \ket{y}\mapsto \ket{x\wedge y}$.

The last few years have seen a push towards generalising graphical calculi to work for higher-dimensional qudits.
For ZX there is now work on qutrits~\cite{wetering2022,wang2018,wang2014qutrit}, the prime-dimensional qudit stabiliser fragment~\cite{booth2022,poor2023qupitstabs}, and the universal algebraic qudit ZX-calculus~\cite{wang2019,wang2021qufinite}.
For ZW there are several different proposals for qudit generalisations~\cite{poor2023completeness,wang2021}.
Missing from these proposals is a generalisation for the ZH-calculus.

In this paper we present for the first time a qudit generalisation of the ZH-calculus. We base this translation on extending the representation of Boolean logic in the qubit ZH generators of~\cite{backens2021} to arithmetic over $\mathbb{Z}_d$. Then the Z- and X-spiders represent respectively the copy $x\mapsto (x,x)$ and addition/negation $(x,y)\mapsto -_d(x+_dy)$, while the H-box represents (up to some Hadamards) the multiplication $(x,y)\mapsto x\cdot_d y$, where the subscript $d$ denotes an operation modulo $d$. This correspondence makes it easy to represent qudit generalisations of Toffoli-like gates. 

In order to motivate this connection, we will first study the qudit generalisation of the Toffoli+Hadamard gate set, which for qubits is known to be computationally universal for quantum circuits~\cite{shi2003}. First, we show that whereas the Toffoli suffices to construct all classical reversible qubit logic circuits, for odd-dimensional qudits we can do the same with the $\ket{0}$-controlled $X$ gate. We find that our construction for these qudit classical reversible circuits from the $\ket{0}$-controlled $X$ gate is asymptotically optimal up to a logarithmic factor. Second, we show that the gate set consisting of the $\ket{0}$-controlled $X$ and Hadamard\footnote{Technically in mathematics a Hadamard matrix is a $\pm1$ matrix of maximum possible determinant, named after Hadamard's 1893 article on the matter \cite{hadamard1893}.  However, we follow the convention of other qudit graphical calculi to refer to the $d$-dimensional Discrete Fourier Transform as the Hadamard \cite{booth2022}.} gates is approximately universal for quantum computing in all odd prime dimensions. Third, we find that phase-free qudit ZH-diagrams represent precisely postselected circuits over this Hadamard+$\ket{0}$-controlled $X$ gate set.

A considerable part of the paper is devoted to proving that the phase-free ZH-calculus for prime-dimensional qudits is universal for matrices over $\mathbb{Z}[\omega]$ where $\omega = e^{2\pi i/d}$ is a $d$th root of unity. While proving universality for qubit ZH is straightforward, the qudit case brings several difficulties, since the structure of the matrix of the H-box is a lot more complicated. Our proof involves an encoding of propositional formulae over $\mathbb{Z}_d$ into polynomials and a construction of Pascal's triangle into a matrix.

In Section~\ref{sec:tof-had} we present our results regarding classical reversible dit logic and the $\ket{0}$-controlled $X$ gate. Then in Section~\ref{qudit-zh} we introduce the phase-free qudit ZH-calculus and show its connection to the previously introduced gates. In Section~\ref{universality-1} we extend the calculus to allow labels over arbitrary rings and prove its universality over this ring. Then in Section~\ref{universality-2} we tackle the harder problem of proving universality of the phase-free ZH-calculus.

\section{The qudit Toffoli+Hadamard gate set}\label{sec:tof-had}

In this paper, we let $d$ denote the dimension of our qudits, so that a single wire in a (circuit) diagram corresponds to $\mathbb{C}^d$. Note that many of our results only work if $d$ is an odd prime. We let $\omega:= e^{2\pi i/d}$ denote a $d$th root of unity. Then the qudit Paulis correspond to $Z\ket{a} = \omega^a \ket{a}$ and $X\ket{a} = \ket{a+_d1}$, where we use subscripts on operators like $+_d$ to denote operations modulo $d$. The controlled $X$ gate (CX) then becomes $\ket{x,y} \mapsto \ket{x,x+_dy}$.
The qudit Hadamard acts as $H\ket{x} = \frac{1}{\sqrt{d}}\sum_y \omega^{x\cdot y} \ket{y}$. 
For qubits, we can write the action of the Toffoli as $\ket{x,y,z} \mapsto \ket{x,y, (x\cdot_2 y)+_2z}$. This definition extends straightforwardly to the qudit setting, where we just take the multiplication and addition to be modulo $d$ instead of modulo $2$.
When allowing \emph{zeroed} ancillae, i.e.~qubits prepared in the $\ket{0}$ state, the Toffoli together with the $X$ gate (which acts as the NOT gate) suffice to construct an arbitrary classical reversible logic circuit.
It turns out however that for certain qudit dimensions, just a two-qudit gate suffices to achieve the analogous result.

We define the \emph{$\ket{0}$-controlled $X$ gate} as acting on the computational basis as follows:
\begin{equation}
    \ket{c, t} \ \mapsto\ 
     \begin{cases}
        \ket{c,t+_d1}, \ &\text{if}\ c = 0 \\
        \ket{c,t}, \ &\text{else}
    \end{cases}
\end{equation}
i.e. by applying an $X$ gate to the target iff the control is $\ket{0}$.

Note that the $\ket{0}$-controlled X gate is not Clifford for any prime qudit dimensions except for the qubit case (for which it is a CNOT gate conjugated by NOTs on the control).

\begin{theorem}\label{thm:daryrev}
    For any odd qudit dimension $d$, any $d$-ary classical reversible function $f:\mathbb{Z}_d^n\to \mathbb{Z}_d^n$ on $n$ dits can be constructed by a circuit of $O(d^n n)$ many $\ket{0}$-controlled $X$ gates and $O(n)$ ancillae prepared in the $\ket{0}$ state.
\end{theorem}
\begin{proposition}\label{prop:revatleast}
    For any qudit dimension $d$, there exist $d$-ary classical reversible functions $f:\mathbb{Z}_d^n\to \mathbb{Z}_d^n$ that require at least $O(n d^n/\log n)$ single-qudit and two-qudit gates to construct, even when allowed $\Omega (n)$ ancillae.
\end{proposition}
\noindent We present the proofs of Theorem~\ref{thm:daryrev} and Proposition~\ref{prop:revatleast} in the appendix.

Interestingly, we only need a two-qudit gate---the $\ket{0}$-controlled $X$ gate---to construct any $d$-ary classical reversible gate (i.e. bijective maps of the form $f:\mathbb{Z}_d^n\to \mathbb{Z}_d^n$) with the help of $\ket{0}$ ancillae. Hence, the $\ket{0}$-controlled $X$ gate is universal for all classical reversible logic---generalising to all odd $d$ what the three-qubit Toffoli gate does for $d=2$.
Hence, it makes sense to consider the generalization of the qubit Toffoli+H gate set to be the qudit gate set containing $\ket 0$-controlled $X$ and Hadamard, which by Theorem~\ref{thm:daryrev} generates all possible qudit generalized Toffoli gates (since they are all classically reversible). 

For qubits, adding the Hadamard gate to all the classical reversible gates (which is generated by the Toffoli gate and zeroed ancillae) suffices for approximately universal quantum computation~\cite{shi2003}.
By combining Theorems~\ref{thm:daryrev} and~\ref{thm:XH-universality} we find that this is in fact true in any prime qudit dimension.
\begin{theorem}\label{thm:XH-universality}
    The $\ket{0}$-controlled X gate and the H gate form an approximately universal gate set for qudits of any odd prime dimension.  In other words, permitting the help of ancillae, this gate set can deterministically approximate any qudit computation up to arbitrarily small error.
\end{theorem}
\begin{proof}
    The proof below suffices for the case where the qudit dimension $d$ is a prime $d > 3$. The proof for the $d = 3$ case consists of constructing all the Cliffords as follows, and the metaplectic gate (a single-qutrit non-Clifford gate) which we construct in Appendix~\ref{app:qutrituniversality} similarly to our construction in Ref.~\cite[Section 3]{Glaudell2022qutritmetaplectic}.

    Define the single-qudit gates $Q[i]$ by $Q[i] \ket{j} = \omega^{\delta_{ij}} \ket{j}$ where $\delta_{ij} = 1$ iff $i = j$.
    In~\cite{wang2020-3} it is shown that CX, H, and the $Q[i]$ gates are universal for quantum computing for prime $d>3$; for $d=3$ this generates the Clifford group. It hence suffices to show that our gate set generates these gates.
    Clearly, inputting a zeroed ancilla to the control of the $\ket{0}$-controlled X gate yields the X gate.  From here, the CX gate is easy to build from X and $\ket{0}$-controlled X gates. 
    We can also exactly synthesize the $Q[0]$ gate deterministically (up to an irrelevant global phase) with just $\ket{0}$-controlled X gates, H gates and a zeroed ancilla:
    \begin{equation}
        \scalebox{1.6}{\input{zhtozx/q0.tikz}}
    \end{equation}
    Conjugating by X gates then yields all $Q[i]$ gates.
\end{proof}

\begin{remark}
    Theorem~\ref{thm:daryrev} and Proposition~\ref{prop:revatleast} build upon previous work of some of the authors~\cite{yeh2022}, which showed for qutrits how to explicitly construct any ternary classical reversible gate using $O(3^n)$ $\ket{0}^{\otimes n}$-controlled $X_{01}$ gates\footnote{$X_{01}$ maps $\ket{0}$ and $\ket{1}$ to each other and is identity on all other basis states.} each with gate count polynomial in $n$, and that there exist ternary classical reversible gates requiring at least $O(n 3^n/\log n)$ gates to construct.
    In this work, we generalise these results to any odd qudit dimension $d$ and we additionally find a construction of the $\ket{0}^{\otimes n}$-controlled $X_{01}$ gate using $O(n)$ gates. Combining these results gives us the $O(n d^n)$ gate count construction of $d$-ary classical reversible gates which is hence near asymptotically optimal in gate count up to a $\log n$ factor.
\end{remark}

\begin{remark}
    A recent preprint~\cite{ZiW2023optsynthmultictrlqudit} appearing after submission of this paper independently discovered a version of Lemmas~\ref{lem:qudit-zzcx01}~and~\ref{lem:qudit-linear-trick-gen} for any odd qudit dimension.  They additionally provide a separate $O(n)$ gate count $\ket{0}^{\otimes n}$-controlled $X_{01}$ gate construction applicable to any even qudit dimension.  By generalisation of Ref.~\cite{yeh2022}, they independently derived our Proposition~\ref{prop:revatleast} and a version of our Theorem~\ref{thm:daryrev} which uses more types of gates than just the $\ket{0}$-controlled $X$, but which does work for all qudit dimensions.
\end{remark}

\section{The qudit ZH-Calculus}\label{qudit-zh}

Now let us introduce the qudit ZH-calculus, which allows for graphical reasoning about qudit Toffoli-like gates.
Diagrams will flow from inputs at the bottom, to outputs on the top (but because our generators will be flexsymmetric~\cite{carette2021,carette2021thesis} the orientation of diagrams in this paper will not matter much).

As is the case for the qubit ZH-calculus, the qudit ZH-calculus will consist of string diagrams built out of two types of generators: Z-spiders and H-boxes.
We define these as follows:
\[
        \input{z-spider.tikz} \ :=\  \sum_{i=0}^{d-1} |i\rangle^{\otimes n}\langle i|^{\otimes m},\qquad\qquad
        \input{qudit-h-box.tikz} \ :=\  \frac{1}{\sqrt d}\sum_{i_1,...,i_m,j_1,...,j_n \in \mathbb{Z}_d}\omega^{i_1\cdot...\cdot i_m\cdot j_1\cdot...\cdot j_n} |j_1...j_n\rangle\langle i_1...i_m|.
\]
This matches the qubit-ZH definitions of~\cite{Backens2019ZH}, except that now the sums go from $0$ to $d-1$ instead of from $0$ to $1$, and we use the $d$th root of unity $\omega = e^{2\pi i / d}$ instead of $-1$. Additionally, we have included a normalization factor of $1/\sqrt{d}$ in the definition of the H-box that will prevent some tedious constants from appearing everywhere~\cite[Ap. E]{deBeaudrap2020}.
As a consequence of this choice of normalisation, the 1-input, 1-output phase-free H-box corresponds exactly to the qudit Hadamard $\ket{x}\mapsto \frac{1}{\sqrt{d}}\sum_y \omega^{x\cdot y} \ket{y}$. 
Note that while the matrix of the qubit H-box consists of just 1's, with a single entry equal to $-1$, for qudits the matrix has a more complicated structure, with different powers of $\omega$ appearing throughout the matrix.
In the next section we will also introduce labelled H-boxes, so we will sometimes refer to diagrams containing just the above generators as \emph{phase-free} ZH-diagrams, following~\cite{backens2021}.

Apart from these generators we have the standard structural generators---identity, swap, cup and cap---needed to make a compact-closed PROP.
Note that the qudit Z-spider and H-box satisfy the same symmetries as their qubit counterparts, meaning we get a flexsymmetric PROP~\cite{carette2021,carette2021thesis}:
\ctikzfig{generator-symmetries-phasefree}

\begin{remark}
    Note that the actual choice of $d$th root of unity $\omega=e^{2\pi i/d}$ is not important. We can choose any primitive $d$th root of unity (i.e.~a complex number $\omega$ satisfying $\omega^d=1$ while $\omega^k\neq 1$ for any $0\leq k<d$), and the rest of our results will also go through.
\end{remark}


There are a couple of useful derived generators we will need:
\begin{equation}\label{eq:derived-gen}
    \input{qudit-x-derived.tikz}\qquad\qquad
    \input{qudit-pauli-x.tikz}\qquad\qquad
    \input{zh-scalars.tikz}
\end{equation}
The first of these is the well-known $X$-spider. The second realizes the Pauli $X$ gate, e.g. the map $X\ket{i} =\ket{i+_d 1}$. The last two generators represent the scalars $\sqrt{d}$ and $1/\sqrt{d}$ respectively.

The (derived) generators of the qubit ZH-calculus can be motivated by a correspondence to Boolean logic~\cite[Eq. 5]{backens2021}.
Similarly, our generators turn out to correspond with arithmetic operations over $\mathbb{Z}_d$:
\begin{equation}\label{eq:qudit-correspondence}
    \input{qudit-correspondence.tikz}
\end{equation}
Note that here for multiplication we have a sequence of three Hadamard instead of just the one in the qubit version. This is because for qudits $H^4=\text{id}$, but not $H^2=\text{id}$. Instead we have $H^2\ket{i} = \ket{-_di}$. This map is sometimes called the \emph{antipode} or \emph{dualiser}~\cite{carette2020}, and we will use it throughout the diagrams in this paper. It turns out to also be equal to a single-input, single-output X-spider.

This interpretation gives a straightforward way to represent the Toffoli and the $\ket{0}$-controlled $X$ gate (writing our diagrams here from left-to-right to match circuit notation):

\begin{equation}\label{eq:zcxinzh}
\input{zh-tof.tikz} \qquad \qquad \quad \input{zhtozx/zcx-in-zh.tikz}
\end{equation}
The correctness of the Toffoli construction follows easily from the interpretation given in Eq.~\eqref{eq:qudit-correspondence}. For the other, note that in the first step we use the trick that a gate controlled on some value, followed by its adjoint, is the same thing as controlling the adjoint on all the other values.
Then the correctness of the ZH-diagram follows from Fermat's little theorem: for all $x \in \mathbb{Z}_d$ for $d$ prime, $x^{d-1} = 0$ if $x = 0$ and $x^{d-1} = 1$ otherwise. The full diagram hence adds $1$ if the control is not $0$.

Many of the rules of the qubit ZH-calculus generalise to qudits; see Figure~\ref{fig:ZH-rules}. For their soundness we refer to Appendix~\ref{soundness}.

\begin{figure}[tb]
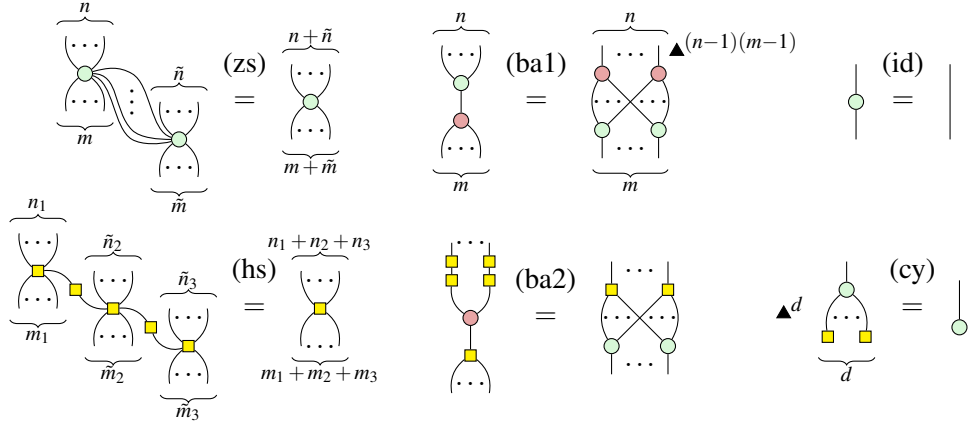

    \ctikzfig{ZH-rules}
    \vspace{-0.25cm}
    \caption{Basic rules of the phase-free qudit ZH-calculus. Some additional (derived) rules are presented in Appendices~\ref{app:further} and~\ref{app:derived}. The rules hold for all $n$ and $m$. Here $d$ is the dimension of the qudit.}
    \label{fig:ZH-rules}
    \vspace{-0.25cm}
\end{figure}

The $Z$-spider fusion rule generalises as expected, but the $H$-box fusion rule generalizes into something that allows \emph{contraction} of odd-length sequences of $H$-boxes interspersed by Hadamards.
For the bialgebra rules, the $Z/X$ version generalises up to global scalars, while the $Z/H$ bialgebra needs some additional Hadamards which would cancel in the qubit case (furthermore for \StrongCompRule, if $(n-1)(m-1) < 0$, introduce $\sd^{-(n-1)(m-1)}$ to the LHS instead).
Lastly, we have the generalization of the \emph{identity} and \emph{multiply} rules. We rename the latter \emph{cyclic} (cy) because what it really captures it the cyclic structure of the group $\mathbb Z_d$.

Note how the above ruleset neither contains a rule stating that $H^4 = \text{id}$, nor an inverted color change rule. That is because we can derive them from the rules presented above:

\begin{equation}
    \label{ft}\tag{h4}\input{qudit-ft-proof-2.tikz}
\end{equation}
\begin{equation}
    \label{h}\tag{h}\input{qudit-color-change-h-proof.tikz}
\end{equation}

Note that in both of these proofs, the application of the bialgebra rule \StrongCompRule does not introduce scalars, as the number of inputs or output of the subdiagram we apply the rule to is always 1. 
We use the name \eqref{h} for the color change rule to keep in line with the notation of~\cite[Fig. 1]{wetering2022}.

Since these derivations hold for arbitrary dimension $d$, they particularly hold for $d=2$. This means that due to the qu\emph{bit} $H$-box fusion rule, \eqref{ft} actually implies the self-inverseness of Hadamard gates, making the (hh) rules of Backens et al's ruleset redundant \cite[Tab. 1]{backens2021}.

In Appendix~\ref{app:further}, we also present a generalisation of the ortho rule from the phase-free qubit ZH-calculus~\cite{backens2021}. Hence, we have a prime-dimensional qudit generalisation of all the phase-free qubit ZH-calculus rewrite rules~\cite{backens2021}.
While those rules are complete for the qubit phase-free calculus, it is not clear whether this continues to hold for qudits. We leave this question for future work, for instance building upon the recent completeness for all qudit dimensions in the ZXW-calculus~\cite{poor2023completeness}.

\subsection{Translating ZH-diagrams to ZX-diagrams}

The qudit ZX-calculus is universal, and hence can represent any linear map between qudits~\cite{wang2014qutrit}.  
So in particular, there must be some way to interpret ZH-diagrams into ZX-diagrams.
As the only generator of ZH-diagrams that is different from the ZX-calculus is the H-box, this is the only one we will have to translate. In fact, we only need to translate the three-legged and one-legged H-box, as the two-legged H-box is just the Hadamard gate.  We can then obtain diagrams for higher-arity H-boxes by unfusing them into three-legged H-boxes.  However, we will also introduce a direct construction for $n$-legged H-boxes, which arises from the asymptotically efficient circuit constructions for any multiple-controlled prime-dimensional qudit Toffoli gate presented in the Appendix.  

First, note that there is a close correspondence between an H-box and the qudit CCZ gate, which acts like $\ket{x,y,z} \mapsto \omega^{x\cdot y\cdot z}\ket{x,y,z}$:
\begin{equation}
\input{zh-ccz.tikz}
\end{equation}
Hence, in particular the three-legged H-box is equal to one copy of the qudit CCZ gate acting on $\ket{\mathnormal +\mathnormal + \mathnormal +}$:
\begin{equation}\label{eq:three-legged-H-box}
    \input{zhtozx/three-legged-H-box.tikz}
\end{equation}
Since a CCZ gate is just the Toffoli from Eq.~\eqref{eq:zcxinzh} with the target qudit conjugated by Hadamards, to construct an H-box in the ZX-calculus it then suffices to show how to construct the qudit Toffoli in the ZX-calculus. But by Theorem~\ref{thm:daryrev} we can construct the Toffoli from the $\ket{0}$-controlled $X$ gate, so that it remains to show how this gate is constructed as a ZX-diagram.

We will write phases on Z-spiders in the ZX-calculus, as vectors $\vec{\alpha}$ of length $d\texttt{-}1$:
\begin{equation}
    \input{zhtozx/Zsp-phase.tikz} \ = \ \ketbra{0\cdots 0}{0\cdots 0} \ +\  e^{i \alpha_1|} \ketbra{1\cdots 1}{1\cdots 1} \ +\  ... \ +\  e^{i \alpha_{d\texttt{-}1}} \ketbra{(d\texttt{-}1)\cdots (d\texttt{-}1)}{(d\texttt{-}1)\cdots (d\texttt{-}1)}
\end{equation}

\begin{lemma}[\cite{yeh2023}]
    The prime-dimensional qudit $\ket{0}$-controlled X gate can be decomposed into the Clifford+Phases gate set (decomposing $H$ as phase gates~\cite[Remark 2.3]{wang2021qufinite}), and written as a ZX-diagram:
    \begin{equation}\label{eq:zx-ket0X}
        \input{zhtozx/zcx.tikz}
    \end{equation}
    where $\vec p = \left(\omega^{\frac{-(d-1)}{2}},\omega^{\frac{-(d-1)}{2}},...,\omega^{\frac{-(d-1)}{2}}\right)$ and
    $\vec r = \left(\omega^{\frac{1}{d}},\omega^{\frac{2}{d}},...,\omega^{\frac{d-1}{d}}\right)$ represents the $d$th root of $Z$ gate from Ref.~\cite{yeh2023}.
\end{lemma}

\begin{theorem}\label{thm:ZH-to-ZX}
    Any prime-dimensional qudit ZH-diagram composed of $m$ Z spiders and $n$ H-boxes each with no more than $g$ legs, can be written as a composition of those $m$ Z spiders, and $O(ng)$ of either $\ket{0}$-controlled $X$, Hadamard or $\ket{0}$.
\end{theorem}
\begin{proof}
    Up to cups and caps, an H-box with two legs is the Hadamard gate, while an H-box with one leg is $Z\ket{\mathnormal +}$:
        $
        \input{zhtozx/one-legged-H-box.tikz} \quad\text{ where } \vec{Z} = \left(\omega,\omega^2,...,\omega^{d\texttt{-}1}\right) \text{ indicates the Z gate.}$

    Eq~\eqref{eq:three-legged-H-box} shows how to relate the three-legged H-box to the CCZ gate. We can then invoke Theorem~\ref{thm:daryrev} to build the Toffoli gate from $\ket{0}$-controlled $X$ gates, which is related to the CCZ gate by conjugating the target by Hadamards. Eq.~\eqref{eq:zx-ket0X} shows how to construct this gate in the ZX-calculus.
    Any H-box with $n > 3$ legs can then be built from unfusing to one- to three-legged H-boxes by applying rule \HPhaseRule. 
\end{proof}

Note that phase-free Z-spiders are just GHZ states $\ket{0...0}+\ket{1...1}+...+\ket{(d\texttt{-}1)...(d\texttt{-}1)}$, up to cups and caps.
Hence, using a typical decomposition of any size qudit GHZ state into a CX circuit on $\ket{\mathnormal +0...0}$,
\begin{equation}
    \input{zhtozx/ghz.tikz}
\end{equation}
we can further decompose the Z-spiders in Theorem~\ref{thm:ZH-to-ZX} into $\{\ket{0}\text{-controlled X},\text{H},\ket{0}\}$ as $\ket{\mathnormal +} = H \ket{0}$.
This then gives us a way to write any phase-free ZH diagram as a $\{\ket{0}\text{-controlled X},\text{H}\}$ circuit where ancillae and postselections on $\ket{0}$ and $\bra{0}$ are allowed.  By Eq.~\eqref{eq:zcxinzh}, we can also write any circuit composed of $\{\ket{0}\text{-controlled X},\text{H},\ket{0},\bra{0}\}$ as a phase-free ZH diagram.  Therefore, prime-dimensional phase-free ZH-diagrams correspond to $\{\ket{0}\text{-controlled X},\text{H}\}$ circuits with ancillae and postselections.
This thus generalizes the same correspondence which holds in the qubit case, between the qubit phase-free ZH-calculus and the Toffoli+Hadamard gate set.

\section{Universality of ZH over arbitrary rings}\label{universality-1}
We will now work towards proving universality of the qudit ZH-calculus for prime dimensions over the ring $\mathbb Z[\omega]$. To do so, it will be helpful to first consider an extended ZH-calculus, where we allow H-boxes labelled by elements of a ring.

So let $R \supset\mathbb Z[\omega, \frac{1}{\sqrt d}]$ be a commutative ring. We now introduce the following additional generators, \emph{labelled} H-boxes:
\begin{equation}
    \label{h-label}\input{qudit-labeled-h-box.tikz}
\end{equation}
Here $r$ is an arbitrary element of $R$. Note that the unlabeled $H$-box corresponds to the $\omega$-labelled one. In writing, we refer to the $\sd$-scaled (1-ary) $r$-labeled $H$-state as $H(r) = (1, r, r^2,...,r^{d-1})^T$. Keeping in line with the notation of Backens et al., we call this calculus ZH$_R$ \cite[Sec. 7]{backens2021}.

The basic idea behind the universality proof is to create a big \emph{Schur product} of simpler matrices. Recall that the Schur product of two matrices $A$ and $B$ of equal dimension is the entrywise product $(A\star B)_{ij} = A_{ij}B_{ij}$.
The Schur product is easily represented in qudit ZH (in the same way as it is for qubits~\cite[p.~27]{backens2021}):
\ctikzfig{zh-schur-product}
We can express an arbitrary $R$-valued matrix $M = (m_{ij})$ as a Schur-product of \emph{$r, 1$-pseudobinary} matrices. These are matrices where every entry of the matrix is either $r$ or $1$. 
Namely, let $\mathcal{R}\subseteq R$ be the, necessarily finite, set of $r\in R$ that appear as entries in $M$. 
Then for $r \in \mathcal{R}$, let $M_r = (m^{(r)}_{ij})$ be the matrix such that $m^{(r)}_{ij} = r$ if $m_{ij}=r$, and $m^{(r)}_{ij}=1$ otherwise. 
Then $M = \bigstar_{r\in\mathcal{R}} M_r$ is the Schur-product of these pseudobinary matrices. To prove universality over a ring $R$ it hence suffices to show that the qudit ZH-calculus can represent arbitrary $r, 1$-pseudobinary matrices for $r\in R$.

In this section we thus introduce the foundational building block of our universality proof: an algorithm for constructing ZH$_R$-diagrams of $r,1$-pseudobinary matrices. For this, we perform two intermediary steps: (1) Describe the location of the ones in a $r,1$-pseudobinary matrix using a logical formula $\varphi$, and (2) convert the formula into a polynomial whose roots are exactly the fulfilling assignments of $\varphi$.

Since we know how to express addition and multiplication as ZH$_R$-diagrams, turning a polynomial into a diagram is then rather straight-forward. The following diagrammatic gadgets, together with those of~\eqref{eq:qudit-correspondence} will prove useful:

\begin{equation}
    \label{useful}\input{mystery-tools-that-will-help-us-later.tikz}
\end{equation}

Consider a linear map $L: (\mathbb C^d)^{\otimes n} \rightarrow (\mathbb C^d)^{\otimes m}$ whose matrix is $r,1$-pseudobinary: for every $\vec x \in \{0,...,d-1\}^n$ we have $L(|\vec x\rangle) = \sum_{\vec y \in \{0,...,d-1\}^m} \lambda_{\vec x, \vec y} |\vec y\rangle$ where all $\lambda_{\vec x, \vec y} \in \{r,1\}$. 
We can describe the location of the $1$s in that matrix using a logical formula $\varphi_L$ in $n+m$ free variables such that $\varphi_L(x_1,...,x_n, y_1,...,y_m)$ is true iff $\lambda_{|x_1...x_n\rangle,|y_1...y_n\rangle} = 1$:
\begin{equation}\label{general-formula}
    \varphi_L(x_1,...,x_n, y_1,...,y_m) = \bigvee_{\substack{i_1,...,i_n,j_1,...,j_m \\\in \{0,...,d-1\}\\\lambda_{i_1...i_n,j_1...j_m}=1}} \bigwedge_{k=1}^n (x_k=i_k) \wedge \bigwedge_{\ell=1}^m (y_\ell = j_\ell).
\end{equation}

Logical formulae unfortunately do not correspond to something we can easily directly express in ZH$_R$. However, we can translate these formulae into polynomials, which we \emph{can} represent in ZH$_R$.

\begin{proposition}\label{formula-to-polynomial}
    If $d$ is prime, then for every propositional formula $\varphi$ over $(\mathbb{Z}_d, -, +, \cdot , =)$ in $n$ free variables there exists a polynomial $p_\varphi \in (\mathbb{Z}_d)[X_1,...,X_n]$ such that 
    $
        p_\varphi(x_1,...,x_n) = 0 \iff \varphi(x_1,...,x_n).
    $
\end{proposition}
\begin{proof}
    Let $\varphi$ be a formula over $(\{1,...,d\}, -,+,\cdot, =)$ in $n$ free variables. We describe our polynomial $p$ inductively. Note that every arithmetic expression in our formula is already a polynomial, since we only allow addition, negation and multiplication in our signature. Thus, we only have to deal with equality, negation and disjunction%
    \footnote{We do not need to deal with conjunction, since $\neg$ and $\vee$ are functionally complete.}.
    \begin{enumerate}[label=\arabic*)]
        \item When $\varphi = (p_1(x_1,...,x_n) = p_2(x_1,...,x_n))$ for $p_1, p_2 \in (\mathbb Z_d)[X_1,...,X_n]$, set $p_\varphi = p_1-p_2$.
        \item When $\varphi = \neg \varphi'$, set $p_\varphi = 1 - (p_{\varphi'})^{d-1}$.
        \item When $\varphi = \varphi_1 \vee \varphi_2$, set $p_\varphi = p_{\varphi_1} \cdot p_{\varphi_2}$.
    \end{enumerate}
    The only non-obvious part of the construction is the construction for negation. This step follows from the fact that for $d$ prime, exponentiating with $d-1$ in $\mathbb Z_d$ maps $0$ to $0$ and everything else to $1$. Lastly, note that the construction in 3) makes use of the absence of zero-divisors in fields. 
\end{proof}

\begin{lemma}\label{lem:polynomial}
    Assume $d$ is prime. Given a polynomial $p \in (\mathbb{Z}_d)[X_1,...,X_n]$ in $n$ variables, we can construct an $n$-input $0$-output $ZH_R$-diagram that evaluates to $1$ on states $|b_1...b_n\rangle$ such that $p(b_1,...,b_n)= 0$, and to $r$ on all other states.
\end{lemma}
\begin{proof}
    First suppose that we had a diagram implementing the map $|b_1...b_n\rangle \mapsto |p(b_1,...,b_n)\rangle$. 
    For $d$ prime, the map $x\mapsto x^{d-1}$ in $\mathbb{Z}_d$ maps $0$ to $0$, and everything else to $1$. By~\eqref{useful}, we know how to realize this operation as a ZH-diagram. Apply this operation to the output of the diagram implementing the polynomial, and postselect with the effect $H(r)^T = (1, r, r^2, ..., r^{d-1})$. This gives the desired map.
    So let's see how to implement the map $|b_1...b_n\rangle \mapsto |p(b_1,...,b_n)\rangle$. 
    
    We do this by induction on the number of variables $n$.
    If $n=0$, then $p \in \mathbb Z_d$ is a constant, which we know how to realize using~\eqref{useful}.
    Now suppose we know how to construct diagrams for polynomials with $n-1$ variables. By definition of polynomial rings we can abuse notation slightly to write $p \in (\mathbb{Z}_d[X_1,...,X_{n-1}])[X_n]$, e.g. 
        $p = \sum_{i=0}^{k} p_i X_n^i\text{ for }p_0,...,p_{k} \in \mathbb{Z}_d[X_1,...,X_{n-1}].$
    By induction, we have ZH-diagrams realizing $p_0,...,p_{k}$, which we denote by boxes labeled ``$p_0$'' through ``$p_{k}$''. Then a diagram for the desired map can be constructed as follows (using the correspondence to algebraic operations of~\eqref{eq:qudit-correspondence}):
    \begin{equation*}
    \scalebox{0.8}{
    \input{zh-polynomial.tikz} 
    } \qedhere
    \end{equation*}
\end{proof}

\sloppy In light of~\eqref{general-formula}, this means that using a polynomial $p \in (\mathbb Z_d)[X_1,...,X_m, Y_1,...,Y_n]$ such that $p(x_1,...,x_m,y_1,...,y_n) = 0 \iff \varphi(x_1,...,x_m, y_1,...y_n)$, we can use Lemma~\ref{lem:polynomial} and map-state-duality to construct arbitrary $r,1$-pseudobinary linear maps as ZH$_R$-diagrams. 

\begin{corollary}\label{prop:r1pseudobinary}
    For prime $d$, every $r,1$-pseudobinary linear map $L: (\mathbb C^d)^{\otimes n}\rightarrow (\mathbb C^d)^{\otimes m}$ has a qudit ZH$_R$-diagram realizing $L$. 
\end{corollary}
\begin{proof}
    Use~\eqref{general-formula} to construct a formula $\phi(\vec x,\vec y)$ that is true when $\bra{\vec y} L \ket{\vec x} = 1$ and $r$ otherwise. Then use Proposition~\ref{formula-to-polynomial} to transform $\phi$ into a polynomial $p$ that is $0$ when $\phi$ is true, and finally use Lemma~\ref{lem:polynomial} to construct a diagram with $n+m$ inputs that evaluates to $1$ when you input $\ket{\vec x}\otimes \ket{\vec y}$ with $p(\vec x,\vec y) = 0$ and to $r$ on other inputs. Bending the last $m$ wires up to be outputs gives a diagram that is exactly equal to $L$.
\end{proof}

We give a worked out example of this entire procedure in Appendix~\ref{appendix:construct}.

\begin{theorem}\label{thm:universality-ring}
    Let $R \supset\mathbb Z[\omega, \frac{1}{\sqrt d}]$ be a commutative ring. Then ZH$_R$ is universal for matrices over $R$.
\end{theorem}
\begin{proof}
    By Proposition~\ref{prop:r1pseudobinary} we can construct ZH$_R$-diagrams for arbitrary $r,1$-pseudobinary matrices for $r\in R$. By taking Schur products of these matrices, any matrix over $R$ can be realised.
\end{proof}

\section{Universality of the phase-free ZH-calculus}\label{universality-2}

We now set our sights on establishing the universality of the phase-free ZH-calculus, where we are only allowed $\omega$-labelled (i.e.~phase-free) H-boxes, for matrices over the ring $\mathbb{Z}[\omega]$.
We will use the structure of the previous proof, reducing the problem to the ability to construct diagrams for $r,1$-pseudobinary matrices, where now $r\in R=\mathbb{Z}[\omega]$. The only obstacle to using this approach is that in the proof of Lemma~\ref{lem:polynomial} we require a postselection to the state $H(r)$, which we don't a priori have access to. To prove universality of the phase-free ZH-calculus we hence need to show that we can construct diagrams for states of the form $H(r) = (1,r,r^2,\ldots, r^{d-1})^T$ where $r=a_1 + a_2\omega +\ldots + a_{d-1}\omega^{d-1} \in \mathbb{Z}[\omega]$.

Backens \emph{et al.}~\cite{backens2021} established the analogous results in the qubit case: that ZH is universal for integer-valued matrices even without introducing labeled $H$-boxes as new generators. To show this, they construct an equivalent to all the integer labelled $H$-boxes: there is a simple expression with the same linear map as the $H(0)$-box, and there is a successor gadget that increments the label of an arbitrary $H$-box by 1. Construction of negative integers is done by using a negation gadget. We will follow a similar path.

First, we already have a representation of $H(0) = \ket{0}$ (see Eq.~\eqref{useful} and take $k=0$). Our immediate goal is then to construct a successor gadget to increment $H$-box labels. This will give us H-boxes with natural numbers as labels. The other possible labelled H-boxes are then straightforward to construct.

\subsection{The qudit successor gadget}

A successor gadget $S = (s_{ij})_{0\le i,j<d}$ that increments the label of an $H$-box by $1$ has to satisfy the equation $SH(a) = H(a+1)$ for any $a$.
Looking at the definition of the qudit $H$-box, this means the coefficients $s_{ij}$ of $S$ have to satisfy the equations $(a+1)^i = \sum_{j=0}^{d-1}s_{ij}a^j$. To solve this, we recall the binomial theorem, which states that $(a+1)^i = \sum_{j=0}^i \binom ij a^j$. Hence, we must have $s_{ij}=\binom ij$, with the convention $\binom ij = 0$ for $j>i$. This means that the matrix $S$ encodes \emph{Pascal's triangle} in the form of a lower triangular matrix. 

Note that because we already have a representation of $H(0)$, that we can use Lemma~\ref{lem:polynomial} to construct a ZH-diagram for any \emph{binary} matrix: a matrix whose entries are only $0$'s and $1$'s.
Our task then is to construct a ZH-diagram for $S$ using only binary matrices. We achieve this by constructing each row of $S$ individually and then \emph{multiplexing} between them.
To see how this works, first consider the linear map $R: \mathbb C^d \rightarrow \mathbb C^d, |i\rangle \mapsto |i\rangle + |i+1\rangle$. 
One readily verifies that the coefficients of $R^j|0\rangle$ for $0\le j< d$ correspond to the $(j+1)$th row of Pascal's triangle. 
Hence, our desired successor gadget $S$ satisfies the equation $R^j \ket{0} = S^T\ket{j}$. Therefore, we need some way to apply a different power of $R$ to different inputs (and then take the transpose, which is straightforward). To do this we need a \emph{multiplexer}.

Consider the linear map $M: (\mathbb C^d)^{d+1} \rightarrow \mathbb C^d$ defined by
\begin{equation*}
    |x_0...x_{d-1}\rangle \otimes |c\rangle \ \mapsto\  \begin{cases}
                                                           |x_c\rangle & x_j=0\text{ for all } j\ne c
                                                           \\0&\text{otherwise}.
                                                       \end{cases}
\end{equation*}

Let $\ket{\varphi^i} = \sum_{j=0}^{d-1} \lambda_{ij}|j\rangle$ be a collection of states for $0 \leq i < d$ where for all $i$ the $\ket{0}$ coefficient $\lambda_{i0}$ equals $1$. Then for a fixed control value $0\le c <d$ we calculate:
\begin{align*}
    M(\ket{\varphi^0}\otimes...\otimes \ket{\varphi^{d-1}}\otimes |c\rangle) \ &=\  \sum_{j_0=0}^{d-1} \cdots \sum_{j_{d-1}=0}^{d-1} \lambda_{0j_0}\cdots \lambda_{(d-1)j_{d-1}} M(|j_0...j_{d-1}\rangle \otimes |c\rangle) 
    \\&=\  \lambda_{00}\cdots\lambda_{(d-1)0} \sum_{j_c=0}^{d-1} \lambda_{cj_c}|j_c\rangle \ =\ \sum_{j_c=0}^{d-1} \lambda_{cj_c}|j_c\rangle \ = \ \ket{\varphi^c} .
\end{align*}

Hence, $M$ multiplexes between these input states, using $\ket{c}$ as a control. 
As each row of Pascal's triangle starts with 1, the states $R^j\ket{0}$ have the right property. Hence $M(R^0\ket{0}\otimes\cdots\otimes R^{d-1}\ket{0}\otimes \ket{c}) = R^c\ket{0} = S^T\ket{c}$. So by combining $M$ and powers of $R$, and applying some appropriate transposes, we get $S$. 

Both maps $R$ and $M$ are binary, meaning we can realize them as a phase-free ZH-diagram using Proposition~\ref{prop:r1pseudobinary}. We perform this construction for $R$ in Appendix \ref{appendix:construct}, while for $M$ we only outline the first few steps, without actually constructing the diagram, due to its immense size. Using placeholders for $M$ and $R$, we get the following diagram for our successor map $S$:
\[ \scalebox{0.9}{
    \input{qudit-succ.tikz}
    }
\]
We can then realize any integer-labeled $H$-box where the label is non-negative: $H(n) = S^n H(0)$. Combining this with Lemma~\ref{lem:polynomial} means we can already construct arbitrary $\mathbb{N}$-valued matrices. To get all integer labeled $H$-boxes, we construct $-1$ in the next subsection. 

\subsection{Constructing all the labelled H-boxes}

To construct more complicated labelled H-boxes we first realise that by taking the Schur product of two labelled H-boxes $H(a)$ and $H(b)$, we calculate the product of the labels (up to global scalar): $H(a)\star H(b) = \frac{1}{\sqrt{d}}H(a\cdot b)$.
Since we already know how to construct $H(n)$ for any $n\in \mathbb{N}$, and we have the phase-free H-box $H(\omega)$ we can then also construct $H(\omega n)$ for any $n \in \mathbb{N}$. 
The second observation is that the successor gadget $S$ adds $1$ to the label regardless of the label, including non-integers. We can hence construct $S^m H(\omega n) = H(\omega n + m)$. Iterating these two steps we can then build $H(\omega^{d-1} n_1 + \omega^{d-2} n_2 + \cdots + \omega n_{d-1} + n_d)$ where all the $n_j\in \mathbb{N}$.

Recall that for a $d$th root of unity $\omega\neq 1$ we have the identity $\sum_{j=1}^{d-1} \omega^j = -1$. Hence using the above procedure we can also construct $H(-1)$ as $H(-1) = H(\omega + \omega^2 + \cdots + \omega^{d-1})$. Combining this with our construction of $H(\sum_j n_j \omega^j)$ for positive $n_j$ above, we can then construct any $H(\sum_j a_j \omega^j)$ where $a_j \in \mathbb{Z}$. For instance, if we want to construct $H(n_2 \omega^2 - n_1 \omega - n_0)$ we do it with the following sequence of operations:
\begin{align*}
H(n_2) \ &\rightarrow\  H(n_2 \omega) \ \rightarrow\  H(-n_2 \omega) \ \rightarrow\  H(-n_2 \omega + n_1) \ \rightarrow\  H(-n_2\omega^2 + n_1\omega) \\
&\rightarrow\  H(-n_2\omega^2 + n_1\omega + n_0) \ \rightarrow\  H(n_2\omega^2 -n_1\omega - n_0).
\end{align*}


To summarise the whole construction of this section: we started out with the observation that with Lemma~\ref{lem:polynomial} we can represent arbitrary polynomials in phase-free ZH, and in this way represent arbitrary binary matrices (where every entry is either $0$ or $1$).
We then found a way to construct a ``successor gadget'' $S$ that increments the label of an H-box, $SH(a) = H(a+1)$, from building blocks that are binary matrices which we know how to construct. Together with using the Schur product as a multiplication operation for H-box labels, this then allowed us to create H-boxes with arbitrary labels from $\mathbb{Z}[\omega]$. But then we can appeal to the same construction in Proposition~\ref{prop:r1pseudobinary} and Theorem~\ref{thm:universality-ring} to conclude the following:
\begin{theorem}
    The phase-free ZH-calculus for qudits of prime dimension $d$ is universal for matrices over the ring $\mathbb{Z}[\omega]$, where $\omega = e^{2\pi i /d}$ is a $d$th root of unity.
\end{theorem}
Note that phase-free ZH-diagrams in fact represent a slightly larger fragment, corresponding to matrices $\frac{1}{\sqrt{d}^k} M$ where $M$ has entries in $\mathbb{Z}[\omega]$. This is because we have global factors of $\frac{1}{\sqrt{d}}$ that cannot be made `local' inside of the matrix. This is analogous to the qubit result for the phase-free ZH-calculus~\cite[Section~8.3]{backens2021} and the Toffoli+Hadamard circuit fragment~\cite{amy2019}.

\section{Conclusion}

We have introduced a qudit ZH-calculus, and showed how to generalise all the rules of the phase-free qubit calculus. We have established a universality result both for qudit ZH over an arbitrary ring, as well as for the phase-free ZH-calculus. 
We found that phase-free ZH-diagrams correspond to postselected circuits of Hadamard and $\ket{0}$-controlled $X$ gates. We showed that this gate set is approximately universal for qudit computation, and we found an almost asymptotically optimal strategy for compiling classical reversible qudit logic to this gate set.

The most immediate question about our qudit ZH-calculus is whether our generalisation of the qubit phase-free rules remains complete for qudits. It is possible to generalise the unique normal form for qubits from~\cite{backens2021}, but it is far from clear how to prove that we can reduce arbitrary diagrams to this normal form.
Another open question is whether our construction in Theorem~\ref{thm:daryrev} is optimal, or whether it can be improved by a logarithmic factor. An interesting future direction would be to translate to and from the ZXW-calculus~\cite{poor2023completeness} to achieve completeness of the qudit ZH-calculus and thereafter improve the challenging compilation of classical reversible logic in photonic quantum computing~\cite{Felice2023ZXWphotonics}.

\textbf{Acknowledgements}: LY is supported by an Oxford - Basil Reeve Graduate Scholarship at Oriel College with the Clarendon Fund. PR was supported by the German Academic Scholarship Foundation. We thank the anonymous reviewers for their feedback.

\bibliographystyle{eptcs}
\bibliography{references}

\appendix

\section{Building all prime-dimensional classical reversible gates}\label{sec:revappendix}

In this section of the appendix, we derive Theorem~\ref{thm:daryrev}.
First, we define a few gates which we will use.  
When an explicit proof is not given, it is because the construction can then be easily checked for any fixed odd prime dimension, and realising that they only affect the $\{\ket{0},\ket{1}\}$ subspace.

The first is a permutation cycle of length $3$ which maps $\ket{00} \mapsto \ket{01}$, $\ket{01} \mapsto \ket{10}$, and $\ket{10} \mapsto \ket{00}$, and is identity on all other computational basis states:
\begin{lemma}
    \begin{equation}
        \scalebox{1.7}{\input{zhtozx/p3.tikz}}
    \end{equation}
\end{lemma}
\begin{proof}
    Consider the action on a basis state $\ket{x,y}$ where $x,y \in \mathbb{Z}_d$.
    After the first gate, the state is $\ket{x,y+x^{d-1}-1}$.
    After the second gate, the state is $\ket{x+(y+x^{d-1}-1)^{d-1}-1,y+x^{d-1}-1}$.
    After the third gate, the state is $\ket{x+(y+x^{d-1}-1)^{d-1}-1,y+x^{d-1}-\left(x+(y+x^{d-1}-1)^{d-1}-1\right)^{d-1}}$.
    At this point, by case distinctions on $x$ and $y$ being $0$, $1$, or otherwise, we can compute that the bottom output state must be:
    \begin{equation}
        \begin{cases}
            \ket{0}, \ &\text{if}\ x = 0 \ \text{and}\ y = 1 \\
            \ket{1}, \ &\text{if}\ x = 1 \ \text{and}\ y = 0 \\
            \ket{y}, \ &\text{else}
        \end{cases}
    \end{equation}
    From here on, the fourth and final gate can be seen to apply when either both $x \neq 1$ and $y = 0$, or both $x = 0$ and $y = 1$.  Hence the top output state must be:
    \begin{equation}
        \begin{cases}
            \ket{0}, \ &\text{if}\ x = 1 \ \text{and}\ y = 0 \\
            \ket{1}, \ &\text{if}\ x = 0 \ \text{and}\ y = 0 \\
            \ket{x}, \ &\text{else}
        \end{cases}
    \end{equation}
    Upon inspection, the circuit sends $\ket{00} \mapsto \ket{01}$, $\ket{01} \mapsto \ket{10}$, and $\ket{10} \mapsto \ket{00}$, and is identity on all other computational basis states.
\end{proof}
Note that here the $\ket{0}$-controlled $X^\dagger$ gate can be constructed from the $\ket{0}$-controlled $X$ gate, by repeating that one $d-1$ times.

Using this, we can build the $\ket{0}$- and $\ket{1}$-controlled $X_{01}$ gate, where the $X_{01}$ gate is a single-qudit permutation gate that maps $\ket{0} \mapsto \ket{1}$, $\ket{1} \mapsto \ket{0}$, and is identity on all other computational basis states.
\begin{lemma}
    \begin{equation}
        \scalebox{1.7}{\input{zhtozx/Z_OCX01.tikz}}
    \end{equation}
\end{lemma}
Note that here we have a $\Lambda$-controlled gate which for $\Lambda(U)$ implements $\ket{x,y} \mapsto U^x\ket{y}$. In this particular case, $\Lambda(X)\ket{x,y} = \ket{x,x+y}$ is the CX gate, and is Clifford.
\begin{corollary}
    The $X_{01}$ gate can be synthesized, by setting the control qudit as an ancilla in the $\ket{0}$ state.
\end{corollary}

We can then build the $\ket{0}$-controlled $X_{01}$ gate.
\begin{lemma}
    \begin{equation}
        \scalebox{1.7}{\input{zhtozx/ZCX01.tikz}}
    \end{equation}
\end{lemma}

We then utilise the following generalisation to all odd qudit dimensions $d$, for which the qubit analogue is in~\cite[Lemma 7.5]{BarencoA1995elementarygates} and the qutrit analogue is in~\cite[Lemma 5]{yeh2022}:
\begin{lemma}\label{lem:qudit-zzcx01}
    \begin{equation}
        \scalebox{1.7}{\input{zhtozx/ZZCX01.tikz}}
    \end{equation}
    where $\Lambda(X_{01})$ can be further decomposed into controlling $X_{01}$ on all odd computational basis states.
\end{lemma}
\begin{corollary}
    Any of these controls can be changed to any other computational basis by conjugating by X's, or to $\Lambda$ controls by repeating the construction once for each odd computational basis state control.
\end{corollary}
\begin{corollary}
    As we will shortly discuss in Proposition~\ref{prop:anytwocycle}, any permutation can be generated by 2-cycles i.e. permutations exchanging only two elements. For example, the $\ket{00}$-controlled $X$ gate can be obtained by $\ket{00}$-controlling each gate in the decomposition:
    \begin{equation}
        \scalebox{1.7}{\input{zhtozx/x-out-of-x01s.tikz}}
    \end{equation}
\end{corollary}

In previous work~\cite{yeh2022}, a construction was found that has polynomial Clifford+$T$ gate count to decompose any tritstring controlled qutrit Toffoli. It was left open whether there was a better construction with linear gate count.  Specifically, whether it was possible to generalise Gidney's construction (reprinted from~\cite{GidneyC2015largecnots}):
\begin{equation}
    \includegraphics[width=5in]{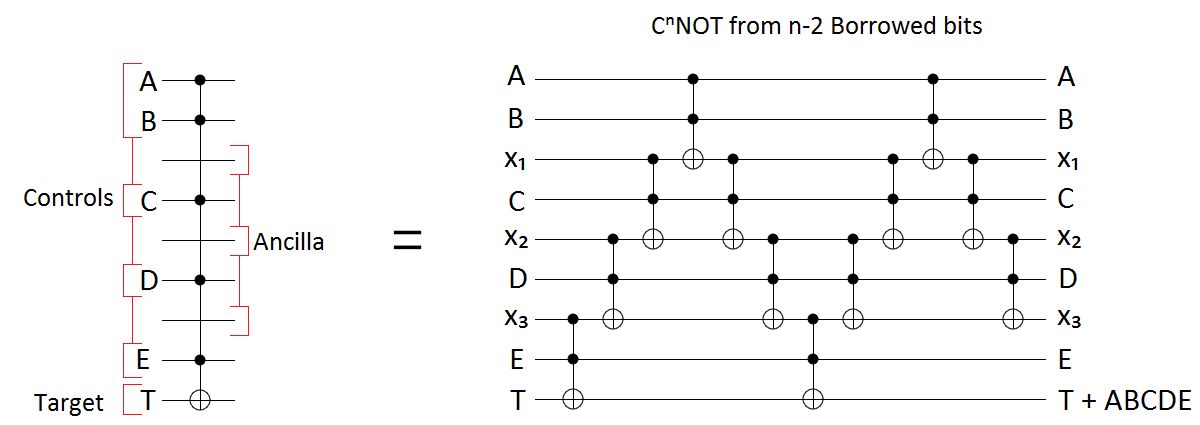}
\end{equation}
If such a construction did not exist, it would be hard to justify ever using qudit Toffolis as opposed to qubit Toffolis, as they would be asymptotically more expensive. However, it turns out that there is a qudit version with an analogous structure.

\begin{lemma}\label{lem:qudit-linear-trick-gen}
    Any odd-dimensional qudit gate controlled on $n$ qudits, admits a decomposition with $n-2$ borrowed ancillae qudits, with $O(n)$ gate count.
\end{lemma}
\begin{proof}
    Any gate $U$ can be controlled on $\ket{0}^{\otimes n}$ with $n-2$ borrowed ancillae:
    \begin{equation}
        \input{figures/zhtozx/qudit-linear-trick-gen.tikz}
    \end{equation}
    This approach is adaptable to arbitrary $n$.  The control qudits can be conjugated by X gates to generalize this from $\ket{0...0}$-controlled, to control on any ditstring.  This also applies if any of the control qudits are other types of controls, for instance $\Lambda$-controlled.

    As an added note, if it is preferred to construct $\Lambda(U)$ instead of $\Lambda(\ket{0}\text{-controlled }U)$, the bottommost control qudit in the above decomposition can be omitted at the cost of adding one borrowed ancilla.
\end{proof}

This lets us immediately apply the following proposition reprinted from Ref.~\cite{yeh2022} for qutrits, which holds for arbitrary qudit dimension.
\begin{proposition}\label{prop:anytwocycle}
    Let $\vec{a} = (a_1,...,a_n)$ and $\vec{b} = (b_1,...,b_n)$ be any two ditstrings of length $n$.
    Then we can exactly implement a unitary which maps the basis states $\ket{\vec{a}} \mapsto \ket{\vec{b}}$ and $\ket{\vec{b}} \mapsto \ket{\vec{a}}$, and is identity on all other computational basis states, with gate count asymptotically the same as the Toffoli controlled on a ditstring of length $n$, which from Lemma~\ref{lem:qudit-linear-trick-gen} is $O(n)$ for any odd qudit dimension.
\end{proposition}
\begin{proof}
    We assume $\vec{a} \neq \vec{b}$, or the permutation 2-cycle $(\vec{a}, \vec{b})$ would just be the identity operation on all inputs.
    As $\vec{a}$ and $\vec{b}$ differ, they must differ by at least one character.  Without loss of generality suppose that $a_n \neq b_n$.
    Consider the following circuit:
    \begin{equation}\label{eq:permutetwoditstrings}
        \input{figures/zhtozx/permute-two-ditstrings.tikz}
    \end{equation}
    Here the circles denote controls on the value of an $a_j$ or $b_j$, which control whether a $X_{a_j,b_j}$ operation is applied (which we take to be the identity if $a_j = b_j$). 
    Hence, the gate in Step~2 is a many-controlled $X_{a_n,b_n}$ gate, which for odd $d$ we know how to build by Lemma~\ref{lem:qudit-linear-trick-gen} using $O(n)$ gates.
    We conjugate this gate, in Steps~1~and~3, by $n-1$ gates that are each Clifford equivalent to the $\ket{0}$-controlled $X_{01}$ gate. Hence for odd $d$, the above circuit requires $O(n)$ gates to implement. 
    
    This circuit indeed implements the $(\vec{a},\vec{b})$ 2-cycle, which we can see by enumerating the possible input cases.
    \begin{itemize}
        \item When the input is $\vec{a}$: Only steps~2~and~3 fire (as $b_n\neq a_n$), outputting $\vec{b}$.
        \item When the input is $\vec{b}$: Steps~1~and~2 fire, outputting $\vec{a}$.
    \end{itemize}
    Observe that when Step~2 does not fire, Steps~1~and~3 always combine to the identity gate.
    Therefore, we only need to consider the remaining cases where Step~2 does fire.
    \begin{itemize}
        \item When both Steps~1~and~2 fired: The input had to have been $\vec{b}$.
        \item When Step~2 fired, but Step~1 didn't fire: Either the input was $\vec{a}$, or the last input character was neither $a_n$ nor $b_n$ in which case the overall operation is the identity.
    \end{itemize}
    Therefore, the circuit in Eq.~\eqref{eq:permutetwoditstrings} maps $\vec{a}$ to $\vec{b}$, $\vec{b}$ to $\vec{a}$, and is identity on all other ditstrings.
\end{proof}

As any permutation can be written as product of disjoint 2-cycles (this is well known, and explained in Ref.~\cite{yeh2022}), we can thus break down any $d$-ary classical reversible circuit on $n$ dits as a permutation of length at most $d^n$, which can always be broken down into a product of $d^n - 1$ 2-cycles.
Thus, we can apply the following qudit equivalent of the qutrit result in Ref.~\cite{yeh2022}:
\begin{theorem}[Restatement of Theorem~\ref{thm:daryrev}]
    For any odd qudit dimension $d$, any $d$-ary classical reversible function $f:\mathbb{Z}_d^n\to \mathbb{Z}_d^n$ on $n$ dits can be constructed by a circuit of $O(d^n n)$ $\ket{0}$-controlled $X$ gates and $O(n)$ ancillae prepared in the $\ket{0}$ state.
\end{theorem}
\begin{proof}
    We view $f$ as a permutation of size $d^n$. This permutation consists of cycles, each of which can be decomposed into 2-cycles. This full decomposition requires at most $d^n-1$ 2-cycles. Implementing each of these 2-cycles requires $O(n)$ gate count.  Therefore, the asymptotic gate count of the overall construction is $O(d^n n)$.
\end{proof}

This is within a log$(n)$ factor of the gate count necessary, by generalizing our proof from Ref.~\cite{yeh2022} to the qudit setting:
\begin{proposition}[Restatement of Proposition~\ref{prop:revatleast}]
    For any qudit dimension $d$, there exist $d$-ary classical reversible functions $f:{\mathbb{Z}_d}^n\to {\mathbb{Z}_d}^n$ that require at least $O(n d^n/\log n)$ single-qudit and two-qudit gates to construct, even when allowed $\Omega (n)$ ancillae.
\end{proposition}
\begin{proof}
    Fix any finite gate set consisting of single-qudit and two-qudit gates, and suppose we have $O(n)$ ancillae. Then taking into account positioning, we have $O(n^2)$ possible single-gate circuits (the square comes from positioning the two-qudit gates). Let's suppose we can bound this by $cn^2$ for some constant $c$.
    Hence, using $N$ gates from this gate set, we can construct at most 
    $(cn^2)^N = c^N n^{2N}$
    different circuits. There are exactly $(d^n)!$ different $d$-ary classical reversible functions on ditstrings of length $n$ (where $k!$ denotes the factorial of $k$). In order to write down every such permutation we must hence have a number of gates $N$ such that at least $c^Nn^{2N} \geq (d^n)!$. Taking the logarithm on both sides and using $\log(k!) \geq \frac12 k\log k$ we can rewrite this inequality to $N\log c + 2N\log n \geq \frac12 d^n\cdot n\log d$. Factoring out $N$ gives $N\geq \frac{\log d}{2} \frac{n d^n}{\log c + 2\log n} \geq \frac{\log d}{6} \frac{n d^n}{\log n}$ for $n\geq c$, which shows that we must have $N=O(n d^n/\log n)$.
\end{proof}
\section{Building the qutrit metaplectic gate}\label{app:qutrituniversality}
This section mostly follows Ref.~\cite[Section 3]{Glaudell2022qutritmetaplectic} to construct the qutrit metaplectic gate $R$, a single-qutrit non-Clifford gate with matrix $\textsf{diag}(1,1,-1)$. However, the key difference is that here we restrict the allowed gate set to H+$\ket{0}$-controlled X instead of qutrit Clifford+$T$.

In all prime qudit dimensions, adding a single-qudit non-Clifford gate to gates generating the Clifford group achieves approximately universal quantum computation~\cite{GottesmanD1999ftqudit}. This was explicitly proven for qutrits in Ref.~\cite[Theorem~2]{CuiS2015universalmetaplectic}.

Therefore, explicit construction of the $R$ gate, in addition to the qutrit Clifford gates constructed in the proof of Theorem~\ref{thm:XH-universality}, suffice to show approximately universality of the H+$\ket{0}$-controlled X gate set.

We remark that in this section only, we match the definition of the $H$ gate in Ref.~\cite{Glaudell2022qutritmetaplectic}, which differs from the definition in the main body of the paper by a global phase of $i$.
In any case, the below construction exactly synthesizes the $R$ gate with the correct global phase, regardless of the global phase in the definition of the $H$ gate.

A code implementation of the below is available at \url{https://github.com/lia-approves/qudit-circuits/blob/main/qupit-Toffoli-Hadamard/ToffH.m}.

By~\cite[Equation 12]{Glaudell2022qutritmetaplectic}, where $\mathbb{I}$ denotes the single-qutrit identity matrix,
\begin{equation}\label{eq:RtwootimesI}
    \begin{tikzpicture}
	\begin{pgfonlayer}{nodelayer}
		\node [style=0 control] (12) at (1.5, 0.75) {\textsf{2}};
		\node [style=box] (13) at (1.5, -0.75) {$-\mathbb{I}$};
		\node [style=none] (14) at (0, -0.75) {};
		\node [style=none] (15) at (0, 0.75) {};
		\node [style=none] (16) at (3, 0.75) {};
		\node [style=none] (17) at (3, -0.75) {};
		\node [style=none] (18) at (-1.25, 0) {$=$};
		\node [style=none] (25) at (4.25, 0) {$=$};
		\node [style=box] (26) at (7, 0.75) {$R$};
		\node [style=none] (27) at (5.5, -0.75) {};
		\node [style=none] (28) at (5.5, 0.75) {};
		\node [style=none] (29) at (8.5, 0.75) {};
		\node [style=none] (30) at (8.5, -0.75) {};
		\node [style=0 control] (31) at (-8.375, 0.75) {\textsf{2}};
		\node [style=box] (32) at (-8.375, -0.75) {$-H^{\dagger}$};
		\node [style=none] (33) at (-9.625, -0.75) {};
		\node [style=none] (34) at (-9.625, 0.75) {};
		\node [style=0 control] (35) at (-4.125, 0.75) {\textsf{2}};
		\node [style=box] (36) at (-4.125, -0.75) {$-H^2$};
		\node [style=none] (37) at (-2.5, 0.75) {};
		\node [style=none] (38) at (-2.5, -0.75) {};
		\node [style=0 control] (40) at (-6.375, 0.75) {\textsf{2}};
		\node [style=box] (41) at (-6.375, -0.75) {$-H^{\dagger}$};
	\end{pgfonlayer}
	\begin{pgfonlayer}{edgelayer}
		\draw (12) to (13);
		\draw (14.center) to (13);
		\draw [style=none] (15.center) to (12);
		\draw (13) to (17.center);
		\draw (16.center) to (12);
		\draw [style=none] (28.center) to (26);
		\draw (29.center) to (26);
		\draw (27.center) to (30.center);
		\draw (31) to (32);
		\draw (33.center) to (32);
		\draw [style=none] (34.center) to (31);
		\draw (35) to (36);
		\draw (36) to (38.center);
		\draw (37.center) to (35);
		\draw (40) to (41);
		\draw (31) to (40);
		\draw (40) to (35);
		\draw (32) to (41);
		\draw (41) to (36);
	\end{pgfonlayer}
\end{tikzpicture}

\end{equation}

By~\cite[Lemma 11]{yeh2022}, where $Z^{\dagger}$ is the conjugation of the $X$ gate by $H$,
\begin{equation}\label{eq:TCS}
    \begin{tikzpicture}
	\begin{pgfonlayer}{nodelayer}
		\node [style=none] (226) at (-9.75, 0) {};
		\node [style=none] (227) at (-9.75, 1.375) {};
		\node [style=none] (228) at (-5.75, 0) {};
		\node [style=none] (229) at (-5.75, 1.375) {};
		\node [style=0 control] (230) at (-7.75, 1.375) {\textsf{2}};
		\node [style=0 control] (231) at (-7.75, 0) {\textsf{2}};
		\node [style=none] (232) at (-4.5, 0) {=};
		\node [style=none] (233) at (8.75, -1.375) {};
		\node [style=box] (234) at (-1, -1.375) {$X_{02}$};
		\node [style=0 control] (235) at (1.5, 1.375) {\textsf{2}};
		\node [style=box] (236) at (1.5, -1.375) {$Z^{\dagger}$};
		\node [style=box] (237) at (4, -1.375) {$X_{02}$};
		\node [style=0 control] (238) at (6.5, 1.375) {\textsf{2}};
		\node [style=box] (239) at (6.5, -1.375) {$Z^{\dagger}$};
		\node [style=none] (240) at (8.75, 1.375) {};
		\node [style=none] (241) at (-3.25, -1.375) {};
		\node [style=none] (242) at (-3.25, 1.375) {};
		\node [style=none] (243) at (-9.75, -1.375) {};
		\node [style=none] (244) at (-5.75, -1.375) {};
		\node [style=0 control] (245) at (1.5, 0) {\textsf{2}};
		\node [style=0 control] (246) at (6.5, 0) {\textsf{2}};
		\node [style=none] (247) at (8.75, 0) {};
		\node [style=none] (248) at (-3.25, 0) {};
		\node [style=none] (249) at (-16.25, 0) {};
		\node [style=none] (250) at (-16.25, 1.375) {};
		\node [style=none] (251) at (-12.25, 0) {};
		\node [style=none] (252) at (-12.25, 1.375) {};
		\node [style=0 control] (253) at (-14.25, 1.375) {\textsf{2}};
		\node [style=box] (254) at (-14.25, 0) {$Q[2]$};
		\node [style=none] (255) at (-11, 0) {=};
		\node [style=none] (256) at (-16.25, -1.375) {};
		\node [style=none] (257) at (-12.25, -1.375) {};
		\node [style=box] (258) at (-7.75, -1.375) {$\omega \mathbb{I}$};
	\end{pgfonlayer}
	\begin{pgfonlayer}{edgelayer}
		\draw (230) to (231);
		\draw (227.center) to (230);
		\draw (230) to (229.center);
		\draw (228.center) to (231);
		\draw (231) to (226.center);
		\draw [style=none] (238) to (240.center);
		\draw [style=none] (241.center) to (234);
		\draw [style=none] (246) to (247.center);
		\draw [style=none] (235) to (245);
		\draw [style=none] (245) to (236);
		\draw [style=none] (238) to (246);
		\draw [style=none] (246) to (239);
		\draw [style=none] (234) to (236);
		\draw [style=none] (236) to (237);
		\draw [style=none] (239) to (237);
		\draw [style=none] (239) to (233.center);
		\draw [style=none, in=180, out=0] (242.center) to (235);
		\draw [style=none] (235) to (238);
		\draw [style=none] (248.center) to (245);
		\draw [style=none] (245) to (246);
		\draw (253) to (254);
		\draw (250.center) to (253);
		\draw (253) to (252.center);
		\draw (251.center) to (254);
		\draw (254) to (249.center);
		\draw [style=none] (256.center) to (257.center);
		\draw (244.center) to (258);
		\draw (258) to (231);
		\draw (258) to (243.center);
	\end{pgfonlayer}
\end{tikzpicture}

\end{equation}

By a modification of~\cite[Lemma 21]{Glaudell2022qutritmetaplectic} to instead use the gate in Equation~\eqref{eq:TCS}, leveraging the decomposition of $H$ into $Z$ and $X$ rotations by~\cite[Remark 2.3]{wang2021qufinite},
\begin{equation}\label{eq:TCminusHdag}
    \begin{tikzpicture}
	\begin{pgfonlayer}{nodelayer}
		\node [style=none] (226) at (-8.5, -0.625) {};
		\node [style=none] (227) at (-8.5, 0.75) {};
		\node [style=none] (240) at (13.5, 0.75) {};
		\node [style=none] (247) at (13.5, -0.625) {};
		\node [style=none] (249) at (-14.5, -0.625) {};
		\node [style=none] (250) at (-14.5, 0.75) {};
		\node [style=none] (251) at (-11, -0.625) {};
		\node [style=none] (252) at (-11, 0.75) {};
		\node [style=0 control] (253) at (-12.75, 0.75) {\textsf{2}};
		\node [style=box] (254) at (-12.75, -0.625) {$-H^{\dagger}$};
		\node [style=none] (255) at (-9.75, 0) {=};
		\node [style=0 control] (256) at (-7, 0.75) {\textsf{2}};
		\node [style=box] (257) at (-7, -0.625) {$Q[2]$};
		\node [style=box] (258) at (-5.25, -0.625) {$H^2$};
		\node [style=0 control] (259) at (-3.5, 0.75) {\textsf{2}};
		\node [style=box] (260) at (-3.5, -0.625) {$Q[2]$};
		\node [style=box] (261) at (-1.75, -0.625) {$H$};
		\node [style=0 control] (262) at (0, 0.75) {\textsf{2}};
		\node [style=box] (263) at (0, -0.625) {$Q[2]$};
		\node [style=box] (264) at (1.75, -0.625) {$H^2$};
		\node [style=0 control] (265) at (3.5, 0.75) {\textsf{2}};
		\node [style=box] (266) at (3.5, -0.625) {$Q[2]$};
		\node [style=box] (267) at (5.25, -0.625) {$H^3$};
		\node [style=0 control] (268) at (7, 0.75) {\textsf{2}};
		\node [style=box] (269) at (7, -0.625) {$Q[2]$};
		\node [style=box] (270) at (8.75, -0.625) {$H^2$};
		\node [style=0 control] (271) at (10.5, 0.75) {\textsf{2}};
		\node [style=box] (272) at (10.5, -0.625) {$Q[2]$};
		\node [style=box] (273) at (12.25, -0.625) {$H^2$};
	\end{pgfonlayer}
	\begin{pgfonlayer}{edgelayer}
		\draw (253) to (254);
		\draw (250.center) to (253);
		\draw (253) to (252.center);
		\draw (251.center) to (254);
		\draw (254) to (249.center);
		\draw (256) to (257);
		\draw (259) to (260);
		\draw (262) to (263);
		\draw (265) to (266);
		\draw (268) to (269);
		\draw (271) to (272);
		\draw (227.center) to (256);
		\draw (256) to (259);
		\draw (259) to (262);
		\draw (262) to (265);
		\draw (265) to (268);
		\draw (268) to (271);
		\draw (271) to (240.center);
		\draw (247.center) to (273);
		\draw (273) to (272);
		\draw (272) to (270);
		\draw (270) to (269);
		\draw (269) to (267);
		\draw (267) to (266);
		\draw (266) to (264);
		\draw (264) to (263);
		\draw (263) to (261);
		\draw (261) to (260);
		\draw (260) to (258);
		\draw (258) to (257);
		\draw (257) to (226.center);
	\end{pgfonlayer}
\end{tikzpicture}

\end{equation}

Finally, we note that for qutrits $X_{12} = -H^2$, enabling us to substitute the $\ket{0}$-controlled $X_{12}$ gate into the $\ket{0}$-controlled $H^2$ gate in Equation~\eqref{eq:RtwootimesI} to realise the qutrit metaplectic gate.


    \section{Soundness of qudit ZH rewrite rules}\label{soundness}
    In this appendix, we argue for the soundness of the rewrite rules introduced in Figure \ref{fig:ZH-rules}. The soundness of the qudit Z-spider fusion rule (zs) is already well known from ZX-literature \cite{booth2022}, and similarly for the identity rule. For the remaining rules, we appeal in parts to the algebraic interpretation of ZH-generators from Eq.~\eqref{eq:qudit-correspondence}, which we will prove along the way. 
    
    First, to verify our claim that the H-box indeed corresponds to multiplication, observe that
    \begin{align*}
       \frac{1}{\sqrt d} \sum_{i=0}^{d-1}\sum_{j=0}^{d-1}\sum_{k=0}^{d-1} \omega^{ijk} H^3|k\rangle\langle i|\langle j| &= \frac{1}{d}\sum_{i=0}^{d-1}\sum_{j=0}^{d-1}\sum_{k=0}^{d-1}\sum_{\nu=0}^{d-1} \omega^{ijk-\nu k}|\nu\rangle\langle i|\langle j|
       = \sum_{i=0}^{d-1}\sum_{j=0}^{d-1}|i \cdot j\rangle\langle i|\langle j|
    \end{align*}
    using the identity $\sum_{k=0}^{d-1} \sum_{\mu=0}^{d-1} (\omega^\mu)^k = d$, since for any root of unity $\zeta \ne 1$ we have $\sum_{k=0}^{d-1} \zeta^k = 0$. Similarly, we get 
    \[
        H^2|i\rangle = \frac 1d\sum_{j=0}^{d-1}\sum_{k=0}^{d-1} \omega^{ ij +jk} |k\rangle = \frac 1d \sum_{k=0}^{d-1} |k\rangle \sum_{j=0}^{d-1} (\omega^{i+k})^j =  |\mathnormal -i\rangle. 
    \]
    Then, using the identity  $H^4 = \id$, we get 
    \[
        \input{h-box-multiply-fusion.tikz}
    \]
    Using flexsymmetry and inductive application of the above identity, we get a generalization of the qubit $H$-fusion rule, with fusion happening via $H^3$ instead of $H$. Only fusion of 1-ary $H$-boxes is not covered. For this, however, observe that 
    \[
        HH(\omega) = \frac 1d \sum_{i=0}^{d-1}\sum_{j=0}^{d-1} \omega^{i+ji}|j\rangle =  |\mathnormal -1\rangle 
    \]
    and thus $H^3H(\omega) = \ket{1}$. Since $1$ is the unit of multiplication modulo $d$, it merges into $H$-boxes. 

    To arrive at our H-box contraction rule, we first introduce the following lemma:
    \begin{lemma}\label{b:h-push-through}
        We can freely transfer a double quantum Fourier transform between the legs of a $H$-box, e.g. 
        \ctikzfig{qudit-h-push-through}
    \end{lemma}
    \begin{proof}
        We have 
        \begin{equation}
            \input{qudit-h-push-through-proof-p1.tikz} = \sum_{i = 0}^{d-1} \sum_{j=0}^{d-1} \sum_{k=0}^{d-1} \omega^{ijk} |\mathnormal-k\rangle \langle i|\langle j| = \sum_{i = 0}^{d-1} \sum_{j=0}^{d-1} \sum_{k=0}^{d-1} \overline\omega^{ijk}|k\rangle \langle i|\langle j| = \sum_{i = 0}^{d-1} \sum_{j=0}^{d-1} \sum_{k=0}^{d-1} \omega^{ijk}|k\rangle \langle i|\langle -j| = \input{qudit-h-push-through-proof-p2.tikz}\label{eq:h-push-through-proof-1}
        \end{equation}
        Using this we then have 
        \begin{equation*}
            \input{qudit-h-push-through-proof-p3.tikz} \qedhere
        \end{equation*}
    \end{proof}

    Then we get:
    
    \ctikzfig{qudit-h-contraction-proof}

    To argue soundness of the two bialgebra rules, we adopt the proof strategy of Backens et al \cite[Eq. 5]{backens2021}, where they reinterpret the qubit ZH-generators as the boolean operators conjunction (and), negation and xor. While our generators no longer correspond to boolean operations, recall from Eq.~\eqref{eq:qudit-correspondence} that we can interpret them as arithmetic operations in $\mathbb Z_d$.
    
    \begin{proof}[Proof of Z/X-bialgebra rule (ba1)]
        To see why the X-spider correspond to addition modulo $d$, we refer to existing literature on ZX-calculus, for instance~\cite{wetering2022}. Using this, we then get
        \ctikzfig{qudit-bialgebra-1-proof}
        Lastly, observe that 
        \[
            \input{qudit-sqrt-d.tikz} = \frac{1}{\sqrt d} \sum_{i=0}^{d-1}\sum_{j=0}^{d-1} \zeta^{ij} = \sqrt d
        \]
        to see that $\sd$ indeed corresponds to $\sqrt d$. 
    \end{proof}
    
    The remaining proofs are straightforward:
    
    \begin{proof}[Proof of Z/H-bialgebra rule (ba2)]
        \begin{equation*}
        \input{qudit-bialgebra-2-proof.tikz} \qedhere
        \end{equation*}
    \end{proof}
    
    \begin{proof}[Proof of cyclic rule (cy)]
        \ctikzfig{qudit-cyclic-proof}
        Here the step marked $(*)$ uses the fact that repeating Pauli-X a total of $d$ times in succession is just the identity. To see our gadget of H-boxes and X-spider really represents Pauli-X, push the two $H$-boxes from the input through to the output to get an addition gadget, and observe that $H^3H(\omega) = |1\rangle$.  
    \end{proof}



    \section{Further Rewrite Rules}\label{app:further}
    
    In the qubit ZH-calculus of Backens \emph{et al.}, there is one further rewrite rule (the \emph{ortho} rule) we have not yet considered in the qudit setting. We make up for that in this appendix. Additionally, we look at two simpler rules \cite[Lem. 2.28, Lem. 5.1]{backens2021} Backens \emph{et al.} have proven to be equivalent to ortho in the qubit setting \cite[Thm. 8.6]{backens2021}, and generalise those as well (though we do not show equivalence). 
    
    The ortho rule (o) is the most complicated qubit rule. It essentially states that 
    \[
        \forall x_0,x_1,y:\  x_0y = x_1(y+1) \iff x_0y = 0 = x_1(y+1).
    \]
    This observation is based on ``exhausting'' all possible values of $\mathbb Z/2\mathbb Z$: No matter what $y \in \mathbb Z/2\mathbb Z$ we choose, we have $\{y, y+1\} = \mathbb Z/2\mathbb Z$. That means either $y$ or $y+1$ is zero, so one of $x_0y$ and $x_1(y+1)$ is always zero. Thus, if $x_0y = x_1(y+1)$, both products must equal 0. 
    
    We can generalise this argument to $\mathbb Z/d\mathbb Z$ for arbitrary $d$: 
    \[
        \forall y:\  \{y,y+1,...,y+d-1\} = \mathbb Z/d\mathbb Z
    \]
    and thus 
    \[
        \forall x_0,...,x_{d-1}, y:\  x_0y = ... = x_{d-1}(y+d-1) \iff \forall i\in\{0,...,d-1\}:\ x_i(y+i) = 0.
    \]
    Expressing this as ZH-diagram (assuming that the Pauli-X inputs are on the left) we get 
    \ctikzfig{qudit-ortho}
    For $d$ prime this rule actually becomes slightly stronger, as the absence of zero-divisors tells us:
    \[
        \forall x_0,...,x_{x-1}, y:\ x_0y = ... = x_{d-1}(y+d-1) \iff \text{at most one }x_i\ne 0,
    \]
    however, this condition does not seem to be easily expressible as a ZH-diagram. 

    Backens \emph{et al.}'s Lemma 2.28 of~\cite{backens2021} states that 
    \[
        \forall x,y:\ xy=1 \iff x = 1 \wedge y = 1.
    \]
    Clearly, this does not hold for $d > 2$, however for prime $d$ we can generalize it to the following statement:
    \[
        \forall x,y:\ xy \ne 0 \iff x\ne 0 \wedge y\ne 0
    \]
    Obviously, this does not hold for $d$ composite, since then $\mathbb Z/d\mathbb Z$ admits zero-divisors. Using Eq.~\eqref{useful} we can realize this as the diagrammatic equation
    \ctikzfig{qudit-no-zero-divisors}
    where all ``$\ldots$'' represent $(d-1)$-fold repetition. 
    
    Lastly, Backens \emph{et al.}'s Lemma 5.1 is a diagrammatic version of the \emph{Frobenius identity}, which states that for prime $d$, we have $x^d=x$ for all $x\in \mathbb Z_d$. We have used statements similar to this all throughout this paper. Diagrammatically, the identity becomes
    \ctikzfig{frobenius-map-identity}
    
    \section{Derived Rewrite Rules}\label{app:derived}
    
    In this appendix we derive simple, yet often useful, rewrite rules using the qudit ZH-calculus. If a rule generalises a known rule for qubit ZH or ZX, we provide references to where the rules were first proven. In these cases, the proofs are often virtually identical and only need some adjustments regarding the number of Hadamards.
    
    Note that Lemmas \ref{push-through} and \ref{push-through-2} also work with the colors of the spiders inverted. Lemma \ref{lem:xs} also works with arbitrarily many wires connecting the spiders, similar to \spiderrule. Lemmas \ref{lem:copy-1} through \ref{lem:copy-4} are copy rules resulting from special cases of our bialgebra rules. Following, we derive some rules for $Z$- and $X$-spiders connected via multiple wires, including the \emph{Z/X Hopf-rule} (Lemma \ref{lem:hopf}), the qudit version of \emph{complementarity} (Lemma \ref{compl}), as well as a generalisation of Yeh and van de Wetering's qutrit \emph{special} rule (Lemma \ref{lem:special}). Lastly, we have some rules about scalar cancellation (Lemma \ref{lem:inner-prod} and \ref{lem:dim}).   
    
    \begin{multicols}{3}
        \begin{lemma}\label{x-id}
            \ctikzfig{qudit-x-id}
        \end{lemma}

        \begin{lemma}\label{x-id-2}
            \ctikzfig{qudit-x-id-2}
        \end{lemma}

        \begin{lemma}\label{push-through}
            \ctikzfig{qudit-push-through}
        \end{lemma}
        
        \begin{lemma}\label{push-through-2}
            \ctikzfig{qudit-push-through-2}
        \end{lemma}

        \begin{lemma}\label{h-push-through}
            \ctikzfig{qudit-h-push-through}
        \end{lemma}

        \begin{lemma}\label{lem:xs}
            \begin{equation}
                \tag{xs}\input{qudit-x-fusion.tikz}
                \label{xs}
            \end{equation}
        \end{lemma}
        
        \begin{lemma}\label{xs-2}
            \ctikzfig{qudit-x-fusion-2}
        \end{lemma}
        
        \begin{lemma}\label{lem:copy-1}
            \ctikzfig{qudit-copy-1}
        \end{lemma}
        
        \begin{lemma}\label{lem:copy-2}
            \ctikzfig{qudit-copy-2}
        \end{lemma}
        
        \begin{lemma}\label{lem:copy-3}
            \ctikzfig{qudit-copy-3}
        \end{lemma}
        
        \begin{lemma}\label{lem:copy-4}
            \ctikzfig{qudit-copy-4}
        \end{lemma}

        \begin{lemma}[Hopf]\label{lem:hopf}
            \begin{equation}
                \label{fake-compl}\tag{H}\input{qudit-fake-complementarity.tikz}
            \end{equation}
        \end{lemma}
        \begin{lemma}\label{compl}
            \ctikzfig{qudit-complementarity}
        \end{lemma}
        \begin{lemma}\label{lem:special}
            \ctikzfig{qudit-additive-negation}
        \end{lemma}
        \begin{lemma}\label{lem:inner-prod}
            \ctikzfig{qudit-lemma-2-4}
        \end{lemma}
        \begin{lemma}\label{lem:dim}
            \begin{equation}
                \label{d}\tag{d}\input{qudit-dim-cancel.tikz}
            \end{equation}
        \end{lemma}
    \end{multicols}

    \renewcommand{\citehack}{\cite[Lem. 2.11]{backens2021}}
    \begin{proof}[Proof of \ref{x-id} \citehack]
        \begin{equation*}
            \input{qudit-x-id-proof.tikz}
        \end{equation*}
    \end{proof}
    
    \begin{proof}[Proof of \ref{x-id-2} \citehack]
        Follows immediately from \eqref{ft}.
    \end{proof}

    \renewcommand{\citehack}{\cite[Lem.~2.15, Lem.~2.16]{backens2021}}
    \begin{proof}[Proof of \ref{push-through} \citehack{} and \ref{push-through-2}]
        Follow by alternating application of \eqref{h}, \XDef and then \eqref{ft}. 
    \end{proof}
    
    \begin{proof}[Proof of \ref{h-push-through}]
        \begin{equation*}
            \input{qudit-h-push-through-pure-proof.tikz} \qedhere
        \end{equation*}
    \end{proof}
    
    \renewcommand{\citehack}{\cite[Lem. 2.10]{backens2021}}
    \begin{proof}[Proof of \ref{lem:xs} \citehack]
        \begin{equation*}
            \input{qudit-x-fusion-proof.tikz}\qedhere
        \end{equation*}
    \end{proof}
    
    \begin{proof}[Proof of \ref{xs-2}]
    \begin{equation*}
            \input{qudit-x-fusion-proof-2.tikz}\qedhere
    \end{equation*}
    \end{proof}
    
    \renewcommand{\citehack}{\cite[Lem. 2.21, Lem. 2.23, Lem. 2.26]{backens2021}}
    \begin{proof}[Proofs of \ref{lem:copy-1} through \ref{lem:copy-4} \citehack]
        Follow immediately from \strongcomprule and \hcomprule.
    \end{proof}
    \begin{proof}[Proof of \ref{lem:hopf}]
        \ctikzfig{qudit-fake-complementarity-proof}
        Note that this proof is very similar to the one given by Feng \cite[Eq. 4.3]{feng2018} and the abstract one done using the dualizer by Duncan and Dunne \cite[Thm. 4.6]{duncan2016}. Lastly, Booth and Carette considered this Lemma to be a rewrite rule \cite[Eq. 21]{booth2022}
    \end{proof}
    \renewcommand{\citehack}{\cite[Lem 2.30]{backens2021}}
    \begin{proof}[Proof of \ref{compl} \citehack]
    \begin{equation*}
            \input{qudit-complementarity-proof.tikz}\qedhere
    \end{equation*}
    \end{proof}
    
    \renewcommand{\citehack}{\cite[Eq. 21]{booth2022}}
    \begin{proof}[Proof of \ref{lem:special} \citehack]
    \begin{equation*}
            \input{qudit-additive-negation-proof.tikz}\qedhere
    \end{equation*}
    \end{proof}
    
    \renewcommand{\citehack}{\cite[Lem. 2.5]{backens2021}}
    \begin{proof}[Proof of \ref{lem:inner-prod} \citehack]
        Unchanged from the qubit case via \hcomprule.
    \end{proof}
    \renewcommand{\citehack}{\cite[Lem. 2.3, Lem 2.4]{backens2021}}
    \begin{proof}[Proof of \ref{lem:dim} \citehack]
    \begin{equation*}
            \input{qudit-dim-cancel-proof.tikz}\qedhere
    \end{equation*}
    \end{proof}
    
    \begin{remark}
        Note that we can also express $\sdi$ using just $Z$- and $X$-spiders in a slightly more intuitive way via 
        \[
            \sdi = \input{qudit-alt-1-over-sqrt-d.tikz}.
        \]
        The proof of this is slightly more complicated. It is based upon the identity 
        \begin{equation}
            \tag{$*$}\input{qudit-dim.tikz}
        \end{equation}
        We then have 
        \ctikzfig{qudit-dim-cancel-catalysis}
        and adding a $\sdi$ to both sides yields the desired identity. Furthermore, Booth and Carette \cite[Fig 1.]{booth2022} give 
        \[
            \sdi = \begin{tikzpicture}
	\begin{pgfonlayer}{nodelayer}
		\node [style=gn] (0) at (0, 0.75) {};
		\node [style=rn] (1) at (0, 0) {};
		\node [style=rn] (2) at (0, -0.75) {};
	\end{pgfonlayer}
	\begin{pgfonlayer}{edgelayer}
		\draw (2) to (1);
		\draw (1) to (0);
		\draw [bend left=45] (0) to (2);
		\draw [bend left=45] (2) to (0);
	\end{pgfonlayer}
\end{tikzpicture}

        \]
        For this version, we have 
        \ctikzfig{qudit-alt-alt-1-over-sqrt-d-proof}
    \end{remark}



    
    \section{Building Diagrams}\label{appendix:construct}
    
    In this appendix, we apply the algorithm outlined in Section \ref{universality-1} to construct a phase-free ZH-diagram for the operator $R$ we used in Section \ref{universality-2} to construct our successor gadget. We also outline how one would go about constructing a phaseless ZH-diagram for $M$, our multiplexer map. 
    
    Recall that we defined $R|i\rangle = |i\rangle + |i+_d1\rangle$, meaning we have 
    \[
        R = \begin{pmatrix}
                1&0&\cdots&0& 1
                \\1&\ddots&\ddots&&0
                \\0&\ddots&\ddots&\ddots&\vdots
                \\\vdots&\ddots&\ddots&\ddots&0
                \\0&\cdots&0&1&1
            \end{pmatrix}.
    \]
    This means that the formula $\varphi_R$ describing the locations of the ones in this matrix is 
    \[
        \varphi_R(x, y) = (y=x) \vee (y=x+_d1).
    \]
    Following the construction of Proposition \ref{formula-to-polynomial}, the polynomial
    \[
        p_R(x,y) = (x-y) \cdot (x+1-y)
    \]
    has roots whenever $\varphi_R(x,y)$ is true. To turn this polynomial into a ZH-diagram realizing $R$, we first realize the map $|x,y\rangle \mapsto |p_R(x,y)\rangle$ via the inductive procedure presented in Lemma \ref{lem:polynomial} to get
    \[
        p_R(x,y) = x^2 + x - xy - yx - y + y^2 = y^2 + (-2x-1)y + (x^2 + x).
    \]
    We thus need to start by constructing diagrams for the polynomials $p_1(x) = -2x-1$ and $p_0(x) = x^2 + x$ (technically, the proof of Lemma \ref{lem:polynomial} starts by constructing diagrams for 0-ary polynomials, e.g. constants such as $-1$ and $-2$. For brevity, we inline this step into the construction of $p_0$ and $p_1$):
    \[
        p_0: \begin{tikzpicture}
	\begin{pgfonlayer}{nodelayer}
		\node [style=none] (0) at (-0.25, -2.5) {};
		\node [style=none] (1) at (-0.25, -3) {$x$};
		\node [style=gn] (2) at (-0.25, -1.5) {};
		\node [style=rn] (4) at (-0.25, 1.5) {};
		\node [style=yellow hadamard] (5) at (-0.25, 2) {};
		\node [style=yellow hadamard] (6) at (-0.25, 2.5) {};
		\node [style=none] (7) at (-0.25, 3) {};
		\node [style=gn] (8) at (0.75, -1.5) {};
		\node [style=yellow hadamard] (9) at (0.75, -0.5) {};
		\node [style=yellow hadamard] (10) at (0.75, 0) {};
		\node [style=yellow hadamard] (11) at (0.75, 0.5) {};
		\node [style=yellow hadamard] (12) at (0.75, 1) {};
		\node [style=rn] (13) at (0.75, 1.5) {};
		\node [style=rn] (14) at (0.25, 1.5) {};
		\node [style=scalar] (15) at (0.25, 1) {};
	\end{pgfonlayer}
	\begin{pgfonlayer}{edgelayer}
		\draw (0.center) to (2);
		\draw (2) to (4);
		\draw (4) to (5);
		\draw (5) to (6);
		\draw (6) to (7.center);
		\draw (2) to (8);
		\draw [bend left] (8) to (9);
		\draw [bend right] (8) to (9);
		\draw (9) to (10);
		\draw (10) to (11);
		\draw (11) to (12);
		\draw (4) to (14);
		\draw (14) to (13);
		\draw (13) to (12);
	\end{pgfonlayer}
\end{tikzpicture}
\qquad\text{and}\qquad p_1:\begin{tikzpicture}
	\begin{pgfonlayer}{nodelayer}
		\node [style=none] (0) at (0, -2.25) {};
		\node [style=none] (1) at (0, -2.75) {$x$};
		\node [style=rn] (3) at (-1, -1.75) {};
		\node [style=rn] (4) at (-1, -1.25) {\tiny$+$};
		\node [style=halfscalar] (5) at (-1.5, -1.75) {};
		\node [style=rn] (6) at (0.75, -2.25) {};
		\node [style=rn] (7) at (0.75, -1.75) {\tiny$+$};
		\node [style=rn] (8) at (0.75, -1.25) {\tiny$+$};
		\node [style=yellow hadamard] (9) at (0.75, 0.25) {};
		\node [style=yellow hadamard] (10) at (0.75, 0.75) {};
		\node [style=yellow hadamard] (11) at (0.75, 1.25) {};
		\node [style=yellow hadamard] (12) at (0.75, 1.75) {};
		\node [style=rn] (13) at (-1, 2.25) {};
		\node [style=yellow hadamard] (14) at (-1, 2.75) {};
		\node [style=yellow hadamard] (15) at (-1, 3.25) {};
		\node [style=none] (16) at (-1, 3.75) {};
		\node [style=none] (18) at (0, -0.5) {};
		\node [style=halfscalar] (19) at (1.25, -2.25) {};
		\node [style=yellow hadamard] (20) at (0.75, -0.75) {};
		\node [style=yellow hadamard] (21) at (0.75, -0.25) {};
		\node [style=yellow hadamard] (22) at (-1, -0.75) {};
		\node [style=yellow hadamard] (23) at (-1, -0.25) {};
		\node [style=rn] (24) at (0.75, 2.25) {};
		\node [style=rn] (25) at (0.25, 2.25) {};
		\node [style=scalar] (26) at (-0.5, 1.75) {};
	\end{pgfonlayer}
	\begin{pgfonlayer}{edgelayer}
		\draw (3) to (4);
		\draw (6) to (7);
		\draw (7) to (8);
		\draw (13) to (14);
		\draw (14) to (15);
		\draw (15) to (16.center);
		\draw (12) to (11);
		\draw (11) to (10);
		\draw (10) to (9);
		\draw [in=-165, out=90] (18.center) to (9);
		\draw (0.center) to (18.center);
		\draw (4) to (22);
		\draw (22) to (23);
		\draw (23) to (13);
		\draw (21) to (9);
		\draw (8) to (20);
		\draw (20) to (21);
		\draw (13) to (25);
		\draw (25) to (24);
		\draw (24) to (12);
	\end{pgfonlayer}
\end{tikzpicture}

    \]
    This then leads to the following diagram for $p_R$
    \[
        \begin{tikzpicture}
	\begin{pgfonlayer}{nodelayer}
		\node [style=none] (0) at (-3, -1) {};
		\node [style=none] (1) at (-3, -1.5) {$x$};
		\node [style=none] (2) at (-2, -1.5) {$y$};
		\node [style=none] (3) at (-2, -1) {};
		\node [style=gn] (4) at (-3, 1) {};
		\node [style=rn] (5) at (-3, 4) {};
		\node [style=yellow hadamard] (6) at (-3, 4.5) {};
		\node [style=yellow hadamard] (7) at (-3, 5) {};
		\node [style=gn] (9) at (-2, 1) {};
		\node [style=yellow hadamard] (10) at (-2, 2) {};
		\node [style=yellow hadamard] (11) at (-2, 2.5) {};
		\node [style=yellow hadamard] (12) at (-2, 3) {};
		\node [style=yellow hadamard] (13) at (-2, 3.5) {};
		\node [style=gn] (15) at (-3, 0.5) {};
		\node [style=none] (17) at (1.5, 0.5) {};
		\node [style=rn] (19) at (0.5, 1) {};
		\node [style=rn] (20) at (0.5, 1.5) {\tiny$+$};
		\node [style=halfscalar] (21) at (0, 1) {};
		\node [style=rn] (22) at (2.25, 0.5) {};
		\node [style=rn] (23) at (2.25, 1) {\tiny $+$};
		\node [style=rn] (24) at (2.25, 1.5) {\tiny$+$};
		\node [style=yellow hadamard] (25) at (2.25, 3) {};
		\node [style=yellow hadamard] (26) at (2.25, 3.5) {};
		\node [style=yellow hadamard] (27) at (2.25, 4) {};
		\node [style=yellow hadamard] (28) at (2.25, 4.5) {};
		\node [style=rn] (29) at (0.5, 5) {};
		\node [style=yellow hadamard] (30) at (0.5, 5.5) {};
		\node [style=yellow hadamard] (31) at (0.5, 6) {};
		\node [style=none] (34) at (1.5, 2.25) {};
		\node [style=halfscalar] (35) at (2.75, 0.5) {};
		\node [style=yellow hadamard] (36) at (2.25, 2) {};
		\node [style=yellow hadamard] (37) at (2.25, 2.5) {};
		\node [style=yellow hadamard] (38) at (0.5, 2) {};
		\node [style=yellow hadamard] (39) at (0.5, 2.5) {};
		\node [style=none] (40) at (-2, -0.5) {};
		\node [style=gn] (41) at (-0.25, 0) {};
		\node [style=yellow hadamard] (42) at (0.5, 6.5) {};
		\node [style=yellow hadamard] (43) at (0.5, 7) {};
		\node [style=yellow hadamard] (44) at (0.5, 7.5) {};
		\node [style=yellow hadamard] (45) at (0.5, 8) {};
		\node [style=none] (46) at (-0.25, 5.75) {};
		\node [style=rn] (47) at (-3, 8.5) {};
		\node [style=rn] (48) at (0.5, 8.5) {};
		\node [style=rn] (49) at (0, 8.5) {};
		\node [style=yellow hadamard] (50) at (-3, 9) {};
		\node [style=yellow hadamard] (51) at (-3, 9.5) {};
		\node [style=none] (52) at (-3, 10) {};
		\node [style=gn] (53) at (3.5, 0) {};
		\node [style=yellow hadamard] (54) at (3.5, 1) {};
		\node [style=yellow hadamard] (55) at (3.5, 1.5) {};
		\node [style=yellow hadamard] (56) at (3.5, 2) {};
		\node [style=yellow hadamard] (57) at (3.5, 2.5) {};
		\node [style=rn] (58) at (3.5, 8.5) {};
		\node [style=rn] (59) at (3, 8.5) {};
		\node [style=scalar] (60) at (1, 8) {};
		\node [style=scalar] (61) at (-2.5, 8) {};
		\node [style=scalar] (62) at (1, 4.5) {};
		\node [style=scalar] (63) at (-2.5, 3.5) {};
		\node [style=rn] (64) at (-2.5, 4) {};
		\node [style=rn] (65) at (-2, 4) {};
		\node [style=rn] (66) at (2.25, 5) {};
		\node [style=rn] (67) at (1.75, 5) {};
		\node [style=none] (68) at (-1.5, 0) {};
	\end{pgfonlayer}
	\begin{pgfonlayer}{edgelayer}
		\draw (4) to (5);
		\draw (5) to (6);
		\draw (6) to (7);
		\draw (4) to (9);
		\draw [bend left] (9) to (10);
		\draw [bend right] (9) to (10);
		\draw (10) to (11);
		\draw (11) to (12);
		\draw (12) to (13);
		\draw (19) to (20);
		\draw (22) to (23);
		\draw (23) to (24);
		\draw (29) to (30);
		\draw (30) to (31);
		\draw (28) to (27);
		\draw (27) to (26);
		\draw (26) to (25);
		\draw [in=-165, out=90] (34.center) to (25);
		\draw (17.center) to (34.center);
		\draw (20) to (38);
		\draw (38) to (39);
		\draw (39) to (29);
		\draw (37) to (25);
		\draw (24) to (36);
		\draw (36) to (37);
		\draw (0.center) to (15);
		\draw (15) to (4);
		\draw (15) to (17.center);
		\draw (3.center) to (40.center);
		\draw (31) to (42);
		\draw (42) to (43);
		\draw (43) to (44);
		\draw (44) to (45);
		\draw (41) to (46.center);
		\draw (7) to (47);
		\draw (47) to (49);
		\draw (49) to (48);
		\draw (48) to (45);
		\draw (47) to (50);
		\draw (50) to (51);
		\draw (51) to (52.center);
		\draw [bend left] (53) to (54);
		\draw [bend right] (53) to (54);
		\draw (41) to (53);
		\draw (54) to (55);
		\draw (55) to (56);
		\draw (56) to (57);
		\draw (48) to (59);
		\draw (59) to (58);
		\draw (58) to (57);
		\draw [in=-165, out=90] (46.center) to (42);
		\draw (5) to (64);
		\draw (64) to (65);
		\draw (65) to (13);
		\draw (29) to (67);
		\draw (67) to (66);
		\draw (66) to (28);
		\draw (68.center) to (41);
		\draw [in=-180, out=90, looseness=1.25] (40.center) to (68.center);
	\end{pgfonlayer}
\end{tikzpicture}

    \]
    Applying $|x\rangle\mapsto|x^{d-1}\rangle$, post-selecting with $H(0)^T$ and bending up the $y$-wire then yields a diagram for $R$. We could now simplify this diagram into something more manageable. Alternatively, we can observe that we do not need to follow the construction of $|x,y\rangle \mapsto |p_R(x,y)\rangle$ given in the proof of Lemma \ref{lem:polynomial}. Often, the structure of a polynomial allows a much simpler ad-hoc construction based on the gadgets given in (\ref{useful}) and (\ref{eq:qudit-correspondence}). In our case, we get 
    \[
        \begin{tikzpicture}
	\begin{pgfonlayer}{nodelayer}
		\node [style=none] (0) at (-0.5, -0.5) {$x$};
		\node [style=none] (1) at (-0.5, 0) {};
		\node [style=none] (2) at (0, 0) {};
		\node [style=none] (3) at (0, -0.5) {$y$};
		\node [style=gn] (4) at (-0.5, 1) {};
		\node [style=gn] (5) at (0, 0.5) {};
		\node [style=rn] (6) at (-0.5, 2.25) {};
		\node [style=yellow hadamard] (7) at (0.25, 4) {};
		\node [style=yellow hadamard] (8) at (0.25, 4.5) {};
		\node [style=yellow hadamard] (9) at (0.25, 5) {};
		\node [style=yellow hadamard] (10) at (0.25, 5.5) {};
		\node [style=none] (11) at (0.25, 6) {};
		\node [style=rn] (12) at (1, 2.25) {};
		\node [style=yellow hadamard] (13) at (-0.5, 2.75) {};
		\node [style=yellow hadamard] (14) at (-0.5, 3.25) {};
		\node [style=yellow hadamard] (15) at (0, 1.25) {};
		\node [style=yellow hadamard] (16) at (0, 1.75) {};
		\node [style=rn] (17) at (1.5, 1.25) {};
		\node [style=rn] (18) at (1.5, 1.75) {\tiny$+$};
		\node [style=yellow hadamard] (19) at (1, 3.25) {};
		\node [style=yellow hadamard] (20) at (1, 2.75) {};
		\node [style=yellow hadamard] (23) at (1, 1.25) {};
		\node [style=yellow hadamard] (24) at (1, 1.75) {};
		\node [style=none] (25) at (0.5, 0.5) {};
		\node [style=none] (27) at (0.5, 1.5) {};
		\node [style=none] (28) at (-1.5, 2.75) {\tiny"$x$-$y$"};
		\node [style=none] (29) at (2.5, 2.75) {\tiny"$x$+$1$-$y$"};
		\node [style=halfscalar] (30) at (2, 1.25) {};
		\node [style=scalar] (31) at (0, 2.25) {};
		\node [style=scalar] (32) at (1.5, 2.25) {};
		\node [style=scalar] (33) at (2, 2.25) {};
		\node [style=none] (34) at (1, 0.75) {};
	\end{pgfonlayer}
	\begin{pgfonlayer}{edgelayer}
		\draw (4) to (6);
		\draw (7) to (8);
		\draw (8) to (9);
		\draw (9) to (10);
		\draw (10) to (11.center);
		\draw (6) to (13);
		\draw (13) to (14);
		\draw (14) to (7);
		\draw (5) to (15);
		\draw (15) to (16);
		\draw (16) to (6);
		\draw (1.center) to (4);
		\draw (2.center) to (5);
		\draw (7) to (19);
		\draw (19) to (20);
		\draw (20) to (12);
		\draw (23) to (24);
		\draw (24) to (12);
		\draw (5) to (25.center);
		\draw [in=-150, out=90] (27.center) to (12);
		\draw (17) to (18);
		\draw (12) to (18);
		\draw (23) to (34.center);
		\draw [in=0, out=-90] (34.center) to (25.center);
		\draw [in=-90, out=0, looseness=1.50] (4) to (27.center);
	\end{pgfonlayer}
\end{tikzpicture}

    \]
    which is significantly easier to simplify, in particular due to the absence of exponentiator gadgets. 
    
    Now, recall our definition of $M$:
    \begin{equation*}
        |x_0...x_{d-1}\rangle \otimes |c\rangle \ \mapsto\  \begin{cases}
                                                            |x_c\rangle & x_j=0\text{ for all } j\ne c
                                                            \\0&\text{otherwise}.
                                                        \end{cases}
    \end{equation*}
    Here, we get the formula 
    \[
        \varphi_M(x_0,...,x_{d-1}, c, y) = \bigvee_{i=0}^{d-1} \left((c = i) \wedge (y = x_i) \wedge \bigwedge_{\substack{j=0\\j\ne i}}^{d-1} (x_j = 0)\right).
    \]
    To translate $\varphi_M$ into a polynomial, we first need to apply deMorgan's law ($A\wedge B= \neg (\neg A \vee \neg B)$) to turn the conjunctions into disjunctions, to get
    \[
        \varphi_M(x_1,...,x_{d-1},c,y) = \bigvee_{i=0}^{d-1} \neg \left((c\ne i) \vee (y\ne x_i) \vee \bigvee_{\substack{j=0\\j\ne i}}^{d-1} (x_j \ne 0)\right)
    \]
    which then translates into the polynomial
    \[
        p_M(x_0,...,x_{d-1},c,y) = \prod_{i=0}^{d-1} \left(1- \left( \left( 1 - (c -i)^{d-1} \right) \cdot \left(1 - (y-x_i)^{d-1}\right) \cdot\prod_{\substack{j=0\\j\ne i}}^{d-1} \left(1-x_j^{d-1}\right)\right)^{d-1}\right).
    \]
    We omit the construction of the associated phase-free ZH-diagram, due to the immense size of the resulting diagram.


\end{document}